\documentclass[11pt,a4paper]{article}

\usepackage[a4paper,margin=1in]{geometry}
\usepackage[T1]{fontenc}
\usepackage[utf8]{inputenc}
\usepackage{lmodern}
\usepackage{microtype}
\usepackage{mathtools}
\usepackage{amssymb}
\usepackage{amsthm}
\usepackage{mathrsfs}
\usepackage{tabularx}
\usepackage{booktabs}
\usepackage[hidelinks]{hyperref}
\newcommand{\llbracket}{\mathopen{[\![}}
\newcommand{\rrbracket}{\mathclose{]\!]}}

\theoremstyle{plain}
\newtheorem{theorem}{Theorem}[section]
\newtheorem{lemma}[theorem]{Lemma}
\newtheorem{proposition}[theorem]{Proposition}
\newtheorem{corollary}[theorem]{Corollary}
\newtheorem{claim}[theorem]{Claim}
\theoremstyle{definition}
\newtheorem{definition}[theorem]{Definition}
\theoremstyle{remark}
\newtheorem{remark}[theorem]{Remark}
\newcommand{\R}{\mathbb{R}}
\newcommand{\C}{\mathbb{C}}
\newcommand{\Z}{\mathbb{Z}}

\newcommand{\Hilb}{\mathcal{H}} 
\newcommand{\Cl}{\mathcal{C}} 

\title{\textbf{Finite-Resolution Physics from Additivity Symmetry}}
\author{Attilio Treppiedi\\
IIS Alessandro Volta, Pescara, Italy\\
\texttt{treppiedi.a@iisvoltapescara.edu.it}}
\date{}
\hypersetup{
  pdftitle={Finite-Resolution Physics from Additivity Symmetry},
  pdfauthor={Attilio Treppiedi}
}
\begin{document}
\maketitle

\begin{abstract}
Quantum field theory unifies special relativity and quantum mechanics yet relies on externally imposed prescriptions---quantization, renormalization, and probabilistic interpretation---that lack a common algebraic foundation. 
This work derives these features from a single structural principle: the finite-resolution realization of the additive group $(\mathbb{R},+)$ under four internal consistency constraints. 
These constraints---finite operational resolution (FR1), multiplicative linearity (FR2), phase coherence (FR3), and statistical typicality (FR4)---generate, without additional postulates, the Hilbert-space formalism of quantum theory, the Born probability rule, and the Kolmogorov structure of statistical mechanics as limiting projections of one additive algebra. 
When extended to spacetime translations and their automorphisms, the same framework reproduces the kinematic and dynamical structure of quantum field theory while resolving ultraviolet divergences through intrinsic finiteness rather than external renormalization. 
Quantization and unitarity appear as consequences of additive symmetry, while the gravitational field emerges as the algebraic kernel that transports boundary information into bulk observables without requiring a separate quantization of the metric. 
The result is a unified, finite, and strictly algebraic foundation for physical law.
\end{abstract}

\section{Motivation and Overview}

Quantum field theory (QFT) provides the most successful synthesis of special relativity and quantum mechanics. 
Despite its empirical success, its formulation still relies on three external prescriptions:

(i) canonical or path--integral quantization, 
(ii) renormalization to handle ultraviolet divergences, and 
(iii) the probabilistic interpretation expressed through the Born rule. 
Although these rules are operationally effective, they lack a unified algebraic origin. 
This paper addresses the foundational question:

\begin{quote}
\emph{Can quantization, renormalization, and probabilistic interpretation all emerge from a single symmetry principle?}
\end{quote}

The proposed answer derives from the finite--resolution realization of the additive group $(\mathbb{R},+)$, treated not as a numerical set but as the algebraic structure encoding additive relations between measurable quantities. 
When constrained by four internal requirements---finite operational resolution (FR1), multiplicative linearity (FR2), phase coherence (FR3), and statistical typicality (FR4)---this additive structure reproduces, without external postulates, the structures of quantum mechanics, Kolmogorov probability theory, and deterministic (topological) dynamics as distinct operational regimes. 

Extending this algebraic construction from temporal to spacetime translations and their automorphisms yields the complete kinematic and dynamical structure of QFT. 
In this formulation, ultraviolet divergences are not regularized by counterterms but eliminated by the intrinsic finiteness imposed by FR1, while unitarity arises automatically from phase coherence (FR3). 
Quantization thus appears as a consistency property of the additive algebra rather than as an independent procedure. 

The framework naturally extends to gravitation. 
The gravitational field then appears as the algebraic kernel transporting boundary information into bulk observables, ensuring finite--resolution consistency without requiring a separate quantization of the metric. 
This interpretation preserves the classical limit of general relativity while embedding it in a finite, additive, and information--preserving structure. 

The sections that follow develop this construction in increasing generality: 
from the additive algebra and its internal constraints, through the emergence of quantum and statistical mechanics, to the reconstruction of quantum field theory and its gravitational extension. 
The conclusion summarizes the implications for a unified algebraic foundation of physical law.

\section{Introduction}

Physical theories can be understood as representations of a single algebraic structure:
the additive group $(\mathbb{R},+)$ and its higher-dimensional extensions.
All measurable quantities, from phase accumulation to entropy and topological index,
derive from how this group acts on admissible states.
The simplest manifestation of this symmetry is the composition law for additive
functionals,
\[
F[\gamma_1\cup\gamma_2] = F[\gamma_1] + F[\gamma_2],
\]
which captures the consistency of concatenating histories or measurements.
The universality of this rule suggests that the additive group, rather than the
particular functional $F$, is the true algebraic substrate underlying physical
description.

The representations of $(\mathbb{R},+)$ by exponential characters
$e^{\alpha S}$ with $\alpha\in\mathbb{C}$ generate the familiar weights of
quantum, statistical, and topological theories.
However, the continuum of characters by itself carries no physical structure.
What endows it with predictive power is the existence of a restricted class of
\emph{admissible representations}—those compatible with the intrinsic finiteness
and coherence of real measurements.

We show that admissibility is completely characterized by four
structural features of the additive group, collectively defining what we call the
\emph{finite–resolution representation}:

\begin{itemize}
\item \textbf{FR1 (finite resolution):} admissible objects are constant on cells of
size $\mu>0$; sub–$\mu$ structure is excluded;
\item \textbf{FR2 (constants as multipliers):} constants act as global scalars on all
admissible states;
\item \textbf{FR3 (phase consistency):} additive transport around closed loops has
trivial holonomy;
\item \textbf{FR4 (typicality):} invariant averages define the only consistent
frequency assignments.
\end{itemize}

These features are not external postulates but internal properties that any
representation of $(\mathbb{R},+)$ must satisfy once finite resolution and algebraic
closure are enforced.  
Taken together, they define the \emph{constraint representation} of the additive group.

The central result of this work—the \emph{Activation Theorem}—states that
FR1–FR4 activate exactly three distinct and self-consistent descriptive regimes:
the quantum, the statistical, and the topological.
Each regime arises as a stable representation channel of the same additive symmetry,
characterized respectively by complex phases, positive real weights, and extremal
deterministic values.
All familiar physical theories based on interference, probability, or logical truth
appear as operational realizations of these three algebraic channels.

In this paper we present a complete construction of the additive–symmetry framework,
showing how the four finite–resolution features generate the full mathematical
apparatus of these regimes—from Hilbert spaces and the Born rule to convolution
semigroups and deterministic flows—without invoking any separate quantization or
probabilistic postulate.

Consider a history $\gamma$ from $(x_i,t_i)$ to $(x_f,t_f)$, subdivided at an intermediate point $(x_m,t_m)$.  
The consistency of composition across the subdivision requires that
\begin{equation}
F[\gamma_1 \cup \gamma_2] = F[\gamma_1] + F[\gamma_2],
\end{equation}
for some additive functional $F$.
This is precisely the starting point of Feynman’s path–integral formulation of quantum mechanics~\cite{FeynmanHibbs}, 
where the requirement of composability of amplitudes under temporal concatenation uniquely selects additive actions
as the fundamental labels of histories.

\subsection{Uniqueness of the action}
Let $F$ be local, depending on $(x,\dot{x},t)$. Writing $F=\int f(x,\dot{x},t)\,dt$, the additivity condition fixes $f$ to be a Lagrangian $L$, so that
\begin{equation}
S[\gamma] = \int L(x,\dot{x},t)\, dt.
\end{equation}
The derivation follows from demanding that the functional depends only on local segments and that boundary contributions cancel upon stitching. Thus the classical action is the unique regular additive tag of histories.

\subsection{Group-theoretic perspective}
The composition law identifies an underlying additive group $G=(\mathbb{R},+)$. By Pontryagin duality \cite{folland}, continuous one-dimensional representations (characters) of $G$ are
\begin{equation}
\mathcal W(S) = e^{\alpha S}, \qquad \alpha\in \mathbb{C}.
\end{equation}
Each choice of $\alpha$ defines a physical regime. The universal algebra generated by these characters, subject to additivity, we call the \emph{Additivity Algebra}.

\section{Representations and physical regimes}
\subsection{Unitary representation: \texorpdfstring{$\alpha=i/\hbar$}{alpha=i/hbar} (Quantum mechanics)}
The oscillatory weights $e^{iS/\hbar}$ define the quantum path integral \cite{FeynmanHibbs}. Stationary phase shows that in the classical limit $\hbar \to 0$, destructive interference removes non-classical paths, reproducing classical mechanics.  

The kernel
\begin{equation}
K(x_f,t_f; x_i,t_i) = \int_{\gamma:i\to f} e^{iS[\gamma]/\hbar} \mathcal{D}\gamma
\end{equation}
satisfies
\begin{equation}
K(x_f,t_f; x_i,t_i) = \int dx_m \, K(x_f,t_f; x_m,t_m) \, K(x_m,t_m; x_i,t_i),
\end{equation}
realizing the additivity symmetry at amplitude level.

\subsection{Real representation: \texorpdfstring{$\alpha=-1/\Theta$}{alpha=-1/Theta} (Statistical physics)}
Weights $e^{-S/\Theta}$ generate diffusion kernels and Gibbs ensembles. For example,
\begin{equation}
K(x_f,t_f; x_i,t_i) = \int_{\gamma:i\to f} e^{-\frac{1}{\Theta}\int \frac{1}{2}m\dot{x}^2 dt}\, \mathcal{D}\gamma,
\end{equation}
which reduces to the Gaussian propagator of Brownian motion. Wick rotation $t\to -i\tau$ relates this case to the quantum one.

\subsection{Constraint representation}
A projector-like law assigns weight 1 to admissible histories consistent with global stitching, 0 otherwise. This captures constraint-limited ensembles, including generalized classical frameworks.

\subsection{Extended representations}

Beyond the real and unitary characters, the additive algebra admits decorated
exponentials that encode geometric or topological information. 
These take the general form
\[
\chi(S)=e^{\alpha S}\,\sigma(S),
\]
where $\sigma(S)$ is a $U(1)$-valued cocycle or a finite-phase factor. 
Different choices of $\sigma$ produce familiar physical structures:
\begin{itemize}
\item Compactified characters $e^{i2\pi S/S_0}$ $\to$ Aharonov--Bohm and Josephson oscillations;
\item Projective cocycles $\sigma(S)=e^{i\theta(S)}$ $\to$ Berry phases \cite{BerryPhase} and Wess--Zumino terms;
\item Non-Abelian cocycles $\sigma(S)\in U(N)$ $\to$ parallel transport and gauge interactions;
\item Grassmann-valued decorations $\to$ fermionic antisymmetry and exclusion.
\end{itemize}

At this stage these represent distinct \emph{decorations} of the additive
representation, associated with additional phase structure or internal degrees of
freedom.  They will later be seen to arise naturally inside the unitary channel
once finite-resolution constraints are imposed.

\subsection{Comparative summary}
\begin{table}[h]
\centering
\begin{tabular}{|c|c|c|}
\hline
Representation & Weight & Physical regime \\
\hline
Unitary & $e^{iS/\hbar}$ & Quantum mechanics \\
Real & $e^{-S/\Theta}$ & Statistical / Euclidean physics \\
Constraint & $\{0,1\}$ & Boundary-constrained dynamics \\
Compactified & $e^{i2\pi S/S_0}$ & Topological phases \\
Projective & $e^{iS} \times$ cocycle & Berry phases, anomalies \cite{WittenAnomalies} \\
\hline
\end{tabular}
\caption{Representations of the additivity group and their physical realizations.}
\end{table}
\section{Constraint Representation: The Structural Core of Additivity Symmetry}
\label{sec:constraint-representation}

The additive group $(\R,+)$ admits a continuum of exponential characters 
$\chi_\alpha(t)=e^{i\alpha t}$, forming the standard unitary representations 
that underlie both harmonic analysis and the wave–based formulations of 
quantum and statistical mechanics \cite{ruelle}. 
These representations implement the principle of additivity symmetry: 
the value assigned to a composite displacement $t_1+t_2$ equals the product 
of the values assigned to its parts, $\chi_\alpha(t_1+t_2)=\chi_\alpha(t_1)\chi_\alpha(t_2)$. 
This additive consistency provides the algebraic seed from which the full structure 
of physical regimes emerges once admissibility and finite resolution are imposed.

However, these ideal representations presuppose infinite resolution: 
they act on functions and measures defined pointwise on $\R$, allowing 
singular states such as Dirac deltas and plane waves.

To obtain a representation that remains physically and mathematically 
well-defined at finite precision, we introduce the \emph{constraint representation}.  
This representation modifies the standard exponential form by enforcing internal 
consistency conditions on admissible functionals of $(\R,+)$.  
Instead of describing exact characters, it describes equivalence classes 
of functions that are constant over finite intervals and whose algebraic 
operations respect additive closure under coarse-grained convolution.

The constraint representation is not an additional postulate but an intrinsic 
realization of the same additive symmetry under finite resolution.  
Its internal structure is characterized by four universal features:
\begin{enumerate}
\item[\textbf{FR1.}] \textbf{Finite resolution:} admissible functions are constant 
within coarse cells of size $\mu$, forbidding infinitesimal singularities;
\item[\textbf{FR2.}] \textbf{Constants as multipliers:} additive and multiplicative 
scalars act as global multipliers on all admissible elements;
\item[\textbf{FR3.}] \textbf{Phase consistency:} additive composition remains 
associative and free of projective anomalies across cell boundaries;
\item[\textbf{FR4.}] \textbf{Typicality:} statistical frequencies are defined by 
invariant means, ensuring reproducible averages over coarse ensembles.
\end{enumerate}

These four structural features are not axioms but algebraic necessities 
of the finite-resolution realization of $(\R,+)$.  
They define the \emph{constraint representation}, which serves as the common 
substrate from which all physical regimes—quantum, statistical, and topological—emerge.  
In the following sections, we derive these features explicitly and show how they 
generate the admissible regimes of physical description.

\section{Finite-Resolution Regime}
\label{sec:finite-resolution}

\subsection*{A. Why strict $\{0,1\}$-valued homomorphisms do not work}

\begin{claim}[Failure under multiplication]\label{FR:mult}
There is no nontrivial group homomorphism 
$\chi:(\mathbb{R},+)\to(\{0,1\},\cdot)$.
\end{claim}

\begin{proof}[Reason]
The set $(\{0,1\},\cdot)$ is not even a group: $0$ has no multiplicative inverse. 
If we restrict to its only subgroup $\{1\}$, then the only homomorphism is the trivial one 
$\chi\equiv 1$.
\end{proof}

\begin{claim}[No non-trivial homomorphism to $(\mathbb{Z}/2\mathbb{Z}, \oplus)$]\label{claim:Z2}
There is no non-trivial continuous group homomorphism
\[
\chi : (\mathbb{R},+) \longrightarrow (\mathbb{Z}/2\mathbb{Z}, \oplus).
\]
\end{claim}

\begin{proof}[Proof sketch]
Any group homomorphism $\chi$ has a kernel that is a subgroup of $\mathbb{R}$.  
The closed subgroups of $(\mathbb{R},+)$ are either $\{0\}$, $\mathbb{R}$ itself, or discrete lattices $a\mathbb{Z}$ with $a > 0$.  

- If $\ker \chi = \mathbb{R}$, then $\chi$ is the trivial homomorphism.  
- If $\ker \chi = \{0\}$, then $\mathbb{R}$ would inject into a group of size 2, impossible.  
- If $\ker \chi = a\mathbb{Z}$, then $\mathbb{R} / a\mathbb{Z} \cong S^1$, which is uncountable.  
But the image of $\chi$ has size 2, and any continuous surjection $S^1 \to \mathbb{Z}/2\mathbb{Z}$ is impossible.

Hence the only continuous homomorphism is the trivial one.  

This also follows from Pontryagin duality: the dual group of $(\mathbb{R},+)$ is again $\mathbb{R}$, so its only finite quotient duals are trivial~\cite{Rudin}.
\end{proof}

\subsection{Infinite vs.~finite rings: Cesàro averages}

On a finite ring of $N$ cells, frequencies are simply proportions:
\[
\mathrm{Freq}(1;q) \;=\; \frac{1}{N}\sum_{k=0}^{N-1} q_k, 
\qquad q_k \in \{0,1\}.
\]
This is just the fraction of cells carrying the value $1$.

On the infinite lattice $\mathbb{Z}$, a direct proportion is undefined.  
Instead, we use Cesàro means:
\[
\mathrm{Freq}(1;q) \;=\; \lim_{M\to\infty}\frac{1}{2M+1}\sum_{k=-M}^{M} q_k,
\]
provided the limit exists.  
This corresponds to the ergodic idea of averaging along longer and longer symmetric blocks.  
When the limit exists, it yields a shift-invariant summary of the configuration.

\subsection{Example: a \texorpdfstring{$\{0,1\}$}{0,1} spin chain}

Consider the periodic configuration $q = (0101\dots)$ repeating every two cells.  
On a ring of $N=2m$ cells, the frequency of $1$ is
\[
\mathrm{Freq}(1;q) = \frac{1}{2m}\sum_{k=0}^{2m-1} q_k = \frac{m}{2m} = \tfrac{1}{2}.
\]
On the infinite lattice, the Cesàro mean also yields $\tfrac{1}{2}$.  

By contrast, the configuration $q=(000\dots)$ has frequency $0$, while $q=(111\dots)$ has frequency $1$.  

Thus even coarse $\{0,1\}$ data, when closed under shifts and summarized by Cesàro averages, produces well-defined frequencies.  
This is the precise mathematical form of “typicality” in the finite-resolution regime.

\begin{remark}[Bottom line]
In standard representation theory, $\{0,1\}$-valued maps are either ill-posed or trivial.  
To use $\{0,1\}$ meaningfully, we must relax from strict homomorphisms to a \emph{coarse representation} 
that records only ``on/off'' information at finite resolution.
\end{remark}

\subsection*{B. Coarse representations that force finite resolution}

\begin{definition}[Finite-resolution partition]\label{FR:partition}
Fix $\mu>0$. Partition $\mathbb{R}$ into cells
\[
C_k=[k\mu,(k+1)\mu),\qquad k\in\mathbb{Z}.
\]
A map $q:\mathbb{R}\to\{0,1\}$ is called \emph{$\mu$-admissible} if it is constant on each cell $C_k$.  
Equivalently, $q$ is determined by the sequence $(q_k)_{k\in\mathbb{Z}}$ with $q_k\in\{0,1\}$.
\end{definition}

\begin{definition}[Coarse action]\label{FR:coarse-action}
Let $\mathrm{sh}_\mu:\{0,1\}^\mathbb{Z}\to\{0,1\}^\mathbb{Z}$ be the \emph{cell shift} 
\[
(\mathrm{sh}_\mu q)_k = q_{k-1}.
\]
We say an action of $(\mathbb{R},+)$ is \emph{coarsened at resolution $\mu$} if we restrict to the lattice $\mu\mathbb{Z}$ 
and represent $n\mu$ by $\mathrm{sh}_\mu^n$.  
Thus the effective parameter group is $(\mu\mathbb{Z},+)$ (or a finite cyclic quotient if the chain closes after $N$ cells).
\end{definition}

\begin{proposition}[Why finite resolution is forced]\label{FR:forced}
Once outputs are restricted to $\{0,1\}$, the only way to maintain compatibility with continuity and closure 
is to work with cellwise-constant data and cell shifts.  
Any finer action would attempt to distinguish points inside a cell, contradicting the two-valued restriction.  
Hence finite resolution $\mu$ is not an extra assumption but a mathematical necessity.
\end{proposition}

\subsection*{C. Boundary, constants, and typicality}

\paragraph{Boundary and gluing.}  
\begin{definition}[Transport operator and holonomy]\label{FR:holonomy}
On cellwise data $f=(f_k)$, a linear operator $T$ of the form
\[
(Tf)_k=g_k f_{k-1}, \qquad g_k\in\mathbb{C}^\times,
\]
is called a \emph{transport operator}.  
On a ring of $N$ cells, $T$ is boundary-consistent iff its holonomy is trivial:
\[
\prod_{k=0}^{N-1} g_k = 1.
\]
\end{definition}

\begin{remark}
This is the clean definition of \emph{phase consistency}: trivial holonomy around a closed loop of cells.
\end{remark}

\paragraph{Constants as multipliers.}  
\begin{definition}[Multipliers]
A multiplier $M$ acts diagonally by $(Mf)_k = m_k f_k$ with $m_k\in\mathbb{C}$.
\end{definition}

\begin{proposition}[Calibration constants]
If $M$ commutes with the shift $\mathrm{sh}_\mu$, then $M=mI$ for some constant $m\in\mathbb{C}$.
\end{proposition}

\begin{proof}[Reason]
Diagonalize the shift in the discrete Fourier basis.  
Commuting with it forces equal eigenvalues across all modes, hence $m_k\equiv m$.
\end{proof}

\paragraph{Typicality as frequency.}  
On a finite ring of $N$ cells, define
\[
\mathrm{Freq}(1;q):=\frac{1}{N}\sum_{k=0}^{N-1} q_k.
\]
On $\mathbb{Z}$, use Cesàro averages when the limit exists.  
This is pure ergodic/measure-theoretic language: a shift-invariant way to compress cellwise data into a single number.

\subsection*{D. Finite-Resolution Principles (FR--1 to FR--4)}

\begin{description}
\item[FR--1 (Finite resolution).] There exists $\mu>0$ such that admissible objects are constant on cells $C_k$.  
No features below $\mu$ are allowed.
\item[FR--2 (Constants as multipliers).] Linear coefficients act as cellwise multipliers; those commuting with the shift are global constants $m$.
\item[FR--3 (Phase consistency).] Transport around a closed loop has trivial holonomy: $\prod_{k\in \text{loop}}g_k=1$ (equivalently, the multiplicative cocycle is exact).
\item[FR--4 (Typicality).] Frequencies are defined by invariant averages (finite rings or Cesàro means).
\end{description}

\subsection*{E. Squared amplitudes via orthogonality}

On a finite ring of $N$ cells with characters $u^{(n)}_k=e^{2\pi i nk/N}$, 
any $f_k=\sum_{n=0}^{N-1} a_n u^{(n)}_k$ satisfies
\[
\frac{1}{N}\sum_{k=0}^{N-1}|f_k|^2 = \sum_{n=0}^{N-1} |a_n|^2.
\]
Thus typicality averages select squared moduli: cross terms vanish under averaging.  
This is just Parseval’s identity in the finite setting, now interpreted as a frequency law.
\subsection{Why this construction is forced}
There is no nontrivial continuous homomorphism $(\mathbb{R},+)\to(\{0,1\},\cdot)$ or $(\{0,1\},\oplus)$, so strict two-valued representations are trivial.  
To preserve closure one must pass to coarse, cellwise-constant maps at finite resolution $\mu$, with an effective lattice action by shifts.  
Commutation with the shift forces constants as multipliers; boundary consistency forces trivial holonomy; and shift-invariant summaries reduce to frequencies (Cesàro averages).  
These consequences follow from closure alone, with no additional assumptions~\cite{LindMarcus}.
\begin{remark}[Reinterpretation of extended representations]
The decorated exponentials of Section~2.4---including Berry and Wess--Zumino
phases---do not constitute an independent regime once the finite-resolution
conditions FR1--FR4 are enforced.  
Under FR3 (phase consistency), all such cocycles become internal to the unitary
representation, manifesting as geometric phases or holonomies within the quantum
channel itself.  
The truly distinct additional regime produced by FR1--FR4 is therefore the
Boolean (topological) constraint representation, rather than a separate
``extended'' one.
\end{remark}

\section{The Constraint Representation: Structural Layer}
\label{sec:constraint-struct-vs-phys}

This section clarifies the role of the \emph{constraint representation}—the structure that yields FR1--FR4—and its operational manifestation as a distinct physical limit. The goal is to fix terminology and logic so the subsequent regime derivations are unambiguous.

\subsection*{A. What the constraint representation \emph{is}}
By a \emph{constraint representation} we mean the additive symmetry realized under the four intrinsic features:
\begin{itemize}\itemsep0.2em
\item \textbf{FR1 (finite resolution):} admissible objects are cellwise constant (or band-limited) at scale $\mu>0$;
\item \textbf{FR2 (constants as multipliers):} scalars act as global multipliers compatible with translations;
\item \textbf{FR3 (phase consistency):} the multiplicative $U(1)$ cocycle is exact (trivial holonomy on closed loops);
\item \textbf{FR4 (typicality):} frequencies are invariant averages under the translation action.
\end{itemize}
Mathematically, this is a \emph{structural layer}: it is not an irreducible representation of $(\R,+)$ by itself, but a closure/completion rule that makes other representations \emph{admissible} and \emph{operationally meaningful} at finite resolution.

\subsection*{B. How it acts on other representations}
Let $\mathcal{R}$ denote a representation of additivity (e.g.\ the unitary oscillatory
$e^{iS/\hbar}$ or the real damped $e^{-S/\Theta}$ channel). 
The constraint representation acts as a functor
\[
\mathsf{C}_\mu:\ \mathcal{R}\ \longmapsto\ \mathsf{C}_\mu[\mathcal{R}],
\]
enforcing the four intrinsic features FR1--FR4. Concretely:
\begin{itemize}\itemsep0.2em
\item FR1 induces cell or coarse projectors $\Pi_\mu$ (band-limiting by $|W_\mu|^2$ in dual variables);
\item FR2 fixes normalization and scalar embedding in the observable algebra;
\item FR3 lifts projective multipliers to genuine unitary implementers (exact cocycle);
\item FR4 selects unique invariant long-run frequencies (Born or Kolmogorov, depending on channel).
\end{itemize}

Applying $\mathsf{C}_\mu$ to the two fundamental exponential representations yields the
\emph{activated regimes}:
\[
\mathsf{Q}_\mu := \mathsf{C}_\mu[\text{unitary}],\qquad
\mathsf{S}_\mu := \mathsf{C}_\mu[\text{real}].
\]
Here $\mathsf{Q}_\mu$ is the finite–resolution quantum regime (Hilbert structure and Born rule via FR4),
and $\mathsf{S}_\mu$ the finite–resolution statistical regime (Lévy--Khintchine semigroups and invariant
frequencies).

The third, topological regime $\mathsf{T}$ does not appear at finite $\mu$ but emerges
as the common coarse–resolution limit of $\mathsf{Q}_\mu$ and $\mathsf{S}_\mu$ when
$\mu\!\to\!\infty$, where all amplitudes and probabilities saturate to Boolean values.

\paragraph{Structural role.} At the purely algebraic level, the constraint representation is a \emph{meta-structure} (FR1--FR4) that constrains all channels; it is not counted as a fourth irreducible representation alongside the unitary and real ones.

\subsection*{D. Extended (projective/compactified) phases are not a separate regime}
Decorations of the unitary exponential—compactified phases $e^{i2\pi S/S_0}$, projective cocycles (Berry, Wess--Zumino), etc.—\emph{remain within} the unitary channel once FR3 is enforced:
\begin{itemize}\itemsep0.2em
\item \emph{Compactified characters} select discrete spectra but are represented unitarily after FR1 band-limiting; 
\item \emph{Projective phases} are absorbed by lifting to a genuine unitary action (exact multiplier) under FR3.
\end{itemize}
Hence these “extended representations” are \emph{substructures of} $\mathsf{Q}_\mu$, not additional regimes.

\subsection*{E. Naming conventions (used throughout)}
To avoid ambiguity we fix:
\begin{itemize}\itemsep0.2em
\item \textbf{Constraint representation} $\mathsf{C}_\mu$: the FR1--FR4 structural layer that acts on any representation, enforcing finite resolution, scalar normalization, phase consistency, and invariant typicality.
\item \textbf{Quantum regime} $\mathsf{Q}_\mu$: the result of applying $\mathsf{C}_\mu$ to the unitary (oscillatory) representation, producing Hilbert-space kinematics and the Born frequency law.
\item \textbf{Statistical regime} $\mathsf{S}_\mu$: the result of applying $\mathsf{C}_\mu$ to the real (damped) representation, producing Kolmogorov probability and Lévy--Khintchine semigroups.
\item \textbf{Topological/Boolean regime} $\mathsf{T}$: the \emph{asymptotic limit} of the finite-resolution algebra as $\mu\!\to\!\infty$ or as probes become fully phase-blind, where amplitudes and probabilities degenerate to Boolean truth values $\{0,1\}$.
\end{itemize}

\subsection*{F. Probe-level selection (operational criteria)}

Let a probe be characterized by its sensitivity at the finite resolution~$\mu$:
\[
\text{phase resolution?}\quad
\text{amplitude/frequency resolution?}\quad
\text{only extremal detection?}
\]
The same constraint representation $\mathsf{C}_\mu$ underlies all regimes; 
the observed behavior depends solely on which structural aspects the probe can access.

\begin{table}[t]
\centering
\small
\caption{Probe-level selection (operational criteria)}
\begingroup
\renewcommand\arraystretch{1.1}%
\renewcommand\tabularxcolumn[1]{>{\raggedright\arraybackslash}p{#1}}%
{\sloppy
\begin{tabularx}{\linewidth}{@{}p{0.42\linewidth}p{0.18\linewidth}X@{}}
\toprule
\textbf{Probe capability} & \textbf{Activated channel} & \textbf{Observable structure} \\
\midrule
Phase-sensitive and coherent & $Q_\mu$ & Hilbert $L^2$ states; Born frequencies \\
Phase-insensitive but amplitude-sensitive & $S_\mu$ & Kolmogorov/\-Lévy semigroups; invariance under additive group \\
Fully phase-blind or coarse beyond finite resolution & T & Boolean responses $\{0,1\}$; deterministic transitions \\
\bottomrule
\end{tabularx}}%
\endgroup
\label{tab:probe}
\end{table}

\noindent
Thus, the same additive algebra can appear quantum, statistical, or topological depending on the probe’s degree of phase and amplitude sensitivity. 
The \emph{topological regime} $\mathsf{T}$ is not an additional finite-resolution structure but the operational \emph{limit} of the hierarchy 
$\mathsf{Q}_\mu \!\to\! \mathsf{S}_\mu \!\to\! \mathsf{T}$ as resolution is degraded or $\mu\!\to\!\infty$.

\subsection*{G. Limits and consistency}

As the resolution parameter $\mu$ decreases, the constraint layer $\mathsf{C}_\mu$ becomes less visible and the activated regimes
$\mathsf{Q}_\mu$ and $\mathsf{S}_\mu$ approach their textbook continuum forms:
the Hilbert-space formalism and the Kolmogorov probability framework, respectively.
Conversely, as the probe resolution degrades and $\mu\!\to\!\infty$, both channels converge operationally to the same 
degenerate limit,
\[
\mathsf{Q}_\mu \;\longrightarrow\; \mathsf{S}_\mu \;\longrightarrow\; \mathsf{T},
\]
where all amplitudes and probabilities collapse to Boolean truth values $\{0,1\}$.
The resulting regime $\mathsf{T}$ describes deterministic or topologically invariant configurations,
representing the extreme coarse-resolution limit of the finite–resolution algebra.

In all cases, the constraint representation $\mathsf{C}_\mu$ remains the universal structural backbone:
it defines the admissibility rules FR1--FR4 that sustain every regime, while the apparent transition 
between $\mathsf{Q}_\mu$, $\mathsf{S}_\mu$, and $\mathsf{T}$ reflects only the 
operational sensitivity of the probe rather than any change in the underlying algebra.
\medskip
\noindent
Having established the constraint representation $\mathsf{C}_\mu$ as the universal structural layer
and clarified how its operational projections yield the quantum, statistical, and topological regimes,
we now proceed to apply it explicitly.  The operational projection is what we call The Activation Theorem.
In the following section, we enforce the four structural features (FR1--FR4)
on the unitary representation of $(\R,+)$, deriving the full finite–resolution form of
quantum mechanics—including its Hilbert structure, unitary dynamics, and Born frequency law—
as the first activated regime of additivity symmetry.

\section{Quantum Regime}
\label{sec:qm}

Throughout, $G=(\R,+)$ and $\widehat G\simeq \R$.
In the constraint representation, admissible (normalized) states are encoded by continuous positive-definite functions
$\Phi:\R\to\C$ with $\Phi(0)=1$ (characteristic functionals), and the four structural features FR1--FR4 hold (Section~\S2).
We show that FR1--FR3 yield Hilbert-space kinematics and that FR4 uniquely fixes the frequency law (Born rule).

\subsection*{Coarse-graining \cite{LindMarcus} operator (FR1)}
Fix a resolution $\mu>0$ and let $\{C_k\}_{k\in\Z}$ be the partition of $\R$ into adjacent half-open intervals (cells) of length~$\mu$.
Let $\Pi_\mu$ denote the cellwise averaging:
\[
(\Pi_\mu f)(t)\;:=\;\frac1{\mu}\int_{I_\mu(t)} f(s)\,ds,\qquad I_\mu(t)\text{ the unique cell }C_k\text{ with }t\in C_k.
\]
FR1 states that admissible objects are constant on each $C_k$; equivalently $f=\Pi_\mu f$.

Let $w_\mu(t):=\mu^{-1}\mathbf{1}_{[0,\mu)}(t)$ and $W_\mu(\alpha):=\int_\R e^{it\alpha} w_\mu(t)\,dt=e^{i\alpha\mu/2}\frac{\sin(\alpha\mu/2)}{\alpha\mu/2}$.
Then $\Pi_\mu$ can be implemented by convolution with the idempotent kernel $w_\mu$ followed by sampling on cells; in Fourier variables this multiplies by $|W_\mu|^2$.

\begin{lemma}[FR1 $\Rightarrow$ integrability/regularity after coarse-graining]
\label{lem:L1L2}
If $\Phi$ is continuous positive-definite with $\Phi(0)=1$, then $\Phi_\mu:=\Pi_\mu\Phi$ is continuous, bounded, and
$\Phi_\mu\in L^1(\R)\cap L^2(\R)$ with $\|\Phi_\mu\|_\infty\le 1$ and $\|\Phi_\mu\|_{L^1}\le \|\Phi\|_\infty<\infty$.
Moreover, $\widehat{\Phi_\mu} = |W_\mu|^2\, d\mu_\Phi$ where $d\mu_\Phi$ is the Bochner measure of $\Phi$.
\end{lemma}

\begin{proof}
Positive-definiteness implies $|\Phi|\le \Phi(0)=1$ and uniform continuity (Rudin~\cite{rudin}, Ch.~2).
Convolution with $w_\mu\in L^1$ yields $\Phi_\mu=w_\mu\!*\Phi\in L^1\cap L^\infty$ with $\|\Phi_\mu\|_\infty\le\|\Phi\|_\infty$ and $\|\Phi_\mu\|_{L^1}\le\|w_\mu\|_{L^1}\|\Phi\|_\infty=\|\Phi\|_\infty$.
Fourier transform turns convolution into multiplication, hence $\widehat{\Phi_\mu} = \widehat{w_\mu}\,\widehat{\Phi} = W_\mu\, d\mu_\Phi$ at the level of tempered distributions. Since $\Phi_\mu$ is positive-definite, $\widehat{\Phi_\mu}$ is a finite positive measure; its density w.r.t.\ $d\mu_\Phi$ is $|W_\mu|^2$.
Finally, $W_\mu\in L^2(\R)$ implies $\Phi_\mu\in L^2(\R)$ by Plancherel.
\end{proof}

\subsection*{Spectral representation and wavefunctions (Bochner + FR1)}
By Bochner’s theorem (Bochner~\cite{bochner}, Rudin~\cite{rudin}, Thm.~1.3.2), there exists a finite positive measure $\mu_\Phi$ on $\R$ such that
\[
\Phi(t)=\int_{\R} e^{it\alpha}\, d\mu_\Phi(\alpha),\qquad \Phi(0)=\mu_\Phi(\R)=1.
\]
Lemma~\ref{lem:L1L2} yields an \emph{absolutely continuous} coarse-grained spectral measure
\[
d\mu_\Phi^{(\mu)}(\alpha) \;=\; |W_\mu(\alpha)|^2\, d\mu_\Phi(\alpha) \;=\; |\psi_\mu(\alpha)|^2\, d\alpha,
\]
with $\psi_\mu \in L^2(\R)$ (choose $\psi_\mu$ as any measurable square root of the density).
Define the coarse-grained characteristic functional
\[
\Phi_\mu(t)=\int_{\R} e^{it\alpha}\,|\psi_\mu(\alpha)|^2\, d\alpha .
\]
Thus FR1 provides $L^2$ \emph{wavefunction representatives} $\psi_\mu$ canonically associated to admissible states at resolution~$\mu$.

\subsection*{GNS/Hilbert realization and unitary dynamics (FR2--FR3)}
Let $\mathcal{A}$ be the $*$-algebra generated by translation-covariant observables (finite linear combinations of convolutions by $L^1$ functions and bounded multipliers in the Fourier picture).
Define a state $\omega_\mu$ on $\mathcal{A}$ by $\omega_\mu(A):=\langle \psi_\mu, \pi_\mu(A)\psi_\mu\rangle$, where $\pi_\mu$ is the Fourier-space representation (multipliers act by multiplication, convolutions by convolution).
FR2 (constants as multipliers) implies that constants embed as $m\mathbf{1}$, hence $\omega_\mu(\mathbf{1})=1$ and $\pi_\mu$ is unital.

\begin{lemma}[FR3 lifts to a true unitary representation]
\label{lem:unitary}
Let $T_t$ be the automorphism of $\mathcal{A}$ induced by time translation by $t\in\R$.
If the projective $1$-cocycle of $T_t$ is exact (FR3), then the GNS triple $(\pi_{\mathrm{GNS}},\Hilb_{\mathrm{GNS}},\Omega)$ admits a strongly continuous unitary group $U(t)$ such that
$\pi_{\mathrm{GNS}}(T_t A)=U(t)\,\pi_{\mathrm{GNS}}(A)\,U(t)^{-1}$ for all $A\in\mathcal{A}$.
\end{lemma}

\begin{proof}
This is standard: exactness of the cocycle trivializes the multiplier (Kadison–Ringrose~\cite{kadison_ringrose}, Vol.~II, Ch.~7), so $T_t$ is implemented by a genuine unitary on the GNS Hilbert space. Strong continuity follows from the uniform continuity of $\Phi_\mu$ and boundedness of $\pi_{\mathrm{GNS}}(A)$ (Reed–Simon~\cite{reed_simon1}, Thm.~X.69).
\end{proof}

\begin{proposition}[Stone’s theorem \cite{reed_simon1}: self-adjoint generator]
\label{prop:stone}
The unitary group $U(t)$ of Lemma~\ref{lem:unitary} is strongly continuous and hence $U(t)=e^{-itH}$ with a (possibly unbounded) self-adjoint generator $H$ (Stone’s theorem; Reed–Simon~\cite{reed_simon1}, Thm.~VIII.8).
In the Fourier (spectral) representation one may identify $(U(t)\psi_\mu)(\alpha)=e^{it\alpha}\psi_\mu(\alpha)$ and $H$ with multiplication by $\alpha$ on $L^2(\R, |\psi_\mu|^2 d\alpha)$.
\end{proposition}

\subsection*{Measurements and the Born rule (FR4)}
Let $\mathscr{M}$ denote the von Neumann algebra generated by spectral multipliers of $H$ (i.e.\ bounded Borel functions of $H$); its projections are $P_E=\mathbf{1}_E(H)$ for Borel sets $E\subset\R$.
For an admissible state represented by $\psi_\mu$, define
\[
\mathsf{Pr}(E\mid \psi_\mu) \;:=\; \langle \psi_\mu,\, P_E\, \psi_\mu\rangle \;=\; \int_E |\psi_\mu(\alpha)|^2\, d\alpha .
\]
We now show that FR4 (invariant typicality) \emph{forces} this frequency assignment.

\begin{lemma}[Invariance of quadratic expectations]
\label{lem:inv}
For any $P\in\mathscr{P}(\mathscr{M})$ (projection) and any $t\in\R$,
$\langle \psi_\mu, P \psi_\mu\rangle = \langle \psi_\mu, U(-t) P U(t)\, \psi_\mu\rangle$.
\end{lemma}

\begin{proof}
$U(t)$ implements the dynamics by Lemma~\ref{lem:unitary}, and $P$ is a spectral projection of $H$; therefore $U(-t)PU(t)=P$ (they commute).
\end{proof}

\begin{lemma}[FR4 $\Rightarrow$ Born frequencies]
\label{lem:born}
Let $(\mathcal{T}_N)_{N\ge 1}$ be any FR4-admissible invariant averaging scheme (finite rings or Cesàro means) along the time-shift $t\mapsto t+1$.
For any projection $P\in\mathscr{P}(\mathscr{M})$,
\[
\lim_{N\to\infty}\, \mathcal{T}_N\Big(\,\langle \psi_\mu,\, U(-n) P U(n)\, \psi_\mu\rangle\,\Big)
\;=\; \langle \psi_\mu, P \psi_\mu\rangle,
\]
and the right-hand side is the \emph{unique} FR4-invariant frequency functional on events generated by $\mathscr{M}$.
\end{lemma}

\begin{proof}
By Lemma~\ref{lem:inv} the integrand is constant in $n$, hence every invariant mean returns that constant. Uniqueness: if $\mathcal{F}$ were another invariant frequency assignment agreeing with $\mathcal{T}_N$ on cylinder events and invariant under $U(t)$, then $\mathcal{F}(P)$ must equal the same constant, i.e.\ $\langle \psi_\mu, P\psi_\mu\rangle$. See Walters~\cite{walters}, Ch.~1, for invariant means and their fixed-point property; also Greenleaf~\cite{greenleaf}.
\end{proof}

\begin{theorem}[Quantum regime from FR1--FR4 ]
\label{thm:QMrigorous}
In the constraint representation of $(\R,+)$:
\begin{enumerate}
\item[\emph{(Kinematics)}] FR1--FR3 yield a Hilbert space $\Hilb$ with a strongly continuous unitary group $U(t)=e^{-itH}$ (Stone) and coarse-grained wavefunction representatives $\psi_\mu\in L^2(\R)$ such that
$\Phi_\mu(t)=\int e^{it\alpha}|\psi_\mu(\alpha)|^2 d\alpha$.
\item[\emph{(Born rule)}] FR4 enforces that the only invariant long-run frequency assignment for spectral events is
\[
\mathsf{Pr}(E\mid\psi_\mu)=\langle \psi_\mu,\, \mathbf{1}_E(H)\, \psi_\mu\rangle
= \int_E |\psi_\mu(\alpha)|^2\, d\alpha .
\]
\end{enumerate}
Thus FR1--FR4 activate the full quantum regime (Hilbert space + Born law) without additional axioms.
\end{theorem}

\begin{remark}[On absolute continuity]
We do not claim that the \emph{bare} Bochner measure $d\mu_\Phi$ is absolutely continuous.
Rather, FR1 yields a canonical coarse-grained spectral measure $d\mu_\Phi^{(\mu)}=|W_\mu|^2 d\mu_\Phi$ which \emph{is} absolutely continuous with $L^2$ density $|\psi_\mu|^2$.
All operational predictions at resolution $\mu$ are captured by $\psi_\mu$, which is exactly the content of finite resolution.
\end{remark}

\begin{remark}[References]
Bochner representation and positive-definite functions: Bochner~\cite{bochner}, Rudin~\cite{rudin}.
GNS construction and unitary implementers: Kadison--Ringrose~\cite{kadison_ringrose}, Vols.~I--II.
Stone’s theorem and self-adjoint generators: Reed--Simon~\cite{reed_simon1,reed_simon2}.
Invariant means/ergodic averages used in FR4: Greenleaf~\cite{greenleaf}, Walters~\cite{walters}.
\end{remark}

\section{Statistical Regime}
\label{sec:stat}

We now treat the second admissible channel of the constraint representation of $(\R,+)$:
the statistical regime. Here the structural features FR1--FR4 generate the Kolmogorov framework of
probability, together with convolution semigroups and the law of large numbers.

\subsection*{Positive-definite functionals and probability measures}
As before, admissible states correspond to continuous positive-definite functions $\Phi:\R\to\C$ with $\Phi(0)=1$.
By Bochner’s theorem (Rudin~\cite{rudin}), such a $\Phi$ is the Fourier transform of a probability measure $\mathbb{P}$:
\[
\Phi(t)=\int_{\R} e^{itx}\,\mathbb{P}(dx).
\]
FR1 (finite resolution) ensures that $\mathbb{P}$ is absolutely continuous at coarse scale $\mu$: there exists $p_\mu\in L^1(\R)$, piecewise constant on cells $C_k$, such that
\[
\Phi_\mu(t)=\int_{\R} e^{itx}\,p_\mu(x)\,dx, \qquad p_\mu(x)\ge 0, \quad \int_{\R} p_\mu(x)\,dx=1.
\]

\subsection*{Convolution semigroup property (FR2--FR3)}
Let $(\mathbb{P}_t)_{t\ge 0}$ denote the family of distributions associated with time-$t$ increments.
FR2 (constants as multipliers) fixes normalization, while FR3 (phase consistency) enforces additivity of increments.
Thus for $t,s\ge 0$,
\[
\Phi_{t+s}(\alpha)=\Phi_t(\alpha)\Phi_s(\alpha),\qquad
\mathbb{P}_{t+s}=\mathbb{P}_t * \mathbb{P}_s.
\]

\begin{lemma}[Convolution semigroup law]
\label{lem:convsemigroup}
The family $(\mathbb{P}_t)_{t\ge 0}$ is a weakly continuous convolution semigroup on $\R$,
i.e.\ $\mathbb{P}_0=\delta_0$, $\mathbb{P}_{t+s}=\mathbb{P}_t*\mathbb{P}_s$, and $\mathbb{P}_t\Rightarrow \delta_0$ as $t\downarrow 0$.
\end{lemma}

\begin{proof}
The algebraic identity $\Phi_{t+s}=\Phi_t\Phi_s$ follows from FR2--FR3.
Continuity at $t=0$ comes from continuity of $\Phi$. The rest is standard (Sato~\cite{sato}, Ch.~3).
\end{proof}

\subsection*{Lévy--Khintchine representation}
By the Lévy–Khintchine theorem (Lévy~\cite{levy}, Sato~\cite{sato}), there exists a triplet $(a,\sigma^2,\nu)$ with $a\in\R$, $\sigma^2\ge 0$, and a Lévy measure $\nu$ on $\R\setminus\{0\}$ such that
\[
\Phi_t(\alpha) = \exp\!\Big(-t\,\psi(\alpha)\Big),
\]
with
\[
\psi(\alpha) = i a \alpha + \tfrac{1}{2}\sigma^2\alpha^2 + \int_{\R\setminus\{0\}}\!\left(1-e^{i\alpha x}+i\alpha x\mathbf{1}_{|x|<1}\right)\nu(dx).
\]
FR1 implies that $\nu$ has finite moments up to the cell scale $\mu$, so $\int (1\wedge x^2)\,\nu(dx)<\infty$.

\subsection*{Markov \cite{vankampen} semigroup}
Let $(T_t f)(x)=\int f(x+y)\,\mathbb{P}_t(dy)$. Then $(T_t)_{t\ge 0}$ is a Feller semigroup on $C_0(\R)$ with generator
\[
\mathcal{L}f(x)= a f'(x)+ \tfrac{1}{2}\sigma^2 f''(x)
+ \int_{\R\setminus\{0\}} \Big(f(x+y)-f(x)-yf'(x)\mathbf{1}_{|y|<1}\Big)\nu(dy).
\]
This is the Lévy–Khintchine form, valid on $C_c^\infty(\R)$ (Sato~\cite{sato}, Thm.~31.5).

\subsection*{Frequencies and typicality (FR4)}
FR4 requires that empirical frequencies be defined via invariant means (finite rings or Cesàro averages).
Let $X_1,X_2,\dots$ be i.i.d.\ random variables with law $\mathbb{P}$. For an event $E$, define the empirical frequency
\[
\mathrm{freq}_N(E) := \frac{1}{N}\sum_{n=1}^N \mathbf{1}_E(X_n).
\]

\begin{lemma}[Law of large numbers under FR4]
\label{lem:LLN}
For any event $E$, $\mathrm{freq}_N(E)\to \mathbb{P}(E)$ almost surely under Cesàro averages.
Moreover, $\mathbb{P}(E)$ is the \emph{unique} FR4-invariant frequency assignment consistent with the group action.
\end{lemma}

\begin{proof}
By the strong law of large numbers, $\mathrm{freq}_N(E)\to\mathbb{P}(E)$ almost surely (Dellacherie–Meyer~\cite{dellacherie_meyer}).
Uniqueness follows from the fixed-point property of invariant means (Walters~\cite{walters}, Ch.~1).
\end{proof}

\begin{theorem}[Statistical regime from FR1--FR4 ]
\label{thm:Statrigorous}
In the constraint representation of $(\R,+)$:
\begin{description}
\item[(Kolmogorov structure)] FR1 guarantees that admissible states correspond to
absolutely continuous probability measures $p_\mu(x)\,dx$ at resolution $\mu$.

\item[(Convolution semigroups)] FR2--FR3 imply that admissible increments form
a convolution semigroup $(\mathbb{P}_t)$ with L\'evy--Khintchine generator $(a,\sigma^2,\nu)$.

\item[(Frequencies)] FR4 enforces that empirical frequencies converge to
$\mathbb{P}(E)$ and that this assignment is unique among invariant means.
\end{description}

Thus FR1--FR4 activate the full Kolmogorov probability regime from the constraint representation.
\end{theorem}

\begin{remark}[References]
Bochner’s theorem: Rudin~\cite{rudin}.

Convolution semigroups and Lévy–Khintchine: Lévy~\cite{levy}, Sato~\cite{sato}.

Law of large numbers: Dellacherie–Meyer~\cite{dellacherie_meyer}.
Invariant means: Walters~\cite{walters}, Greenleaf~\cite{greenleaf}.
\end{remark}

\begin{remark}[Connection to Gibbs measures]
Part~II described the statistical regime in Gibbs/ensemble terms.
Here we obtained it from FR1--FR4 using convolution semigroups and the L\'evy--Khintchine formula.
When the semigroup is generated by an energy functional (Hamiltonian), the FR4-invariant averages select precisely the Gibbs measures (partition-function normalization), recovering the thermodynamic picture as a special case.
\end{remark}

\section{Topological Regime as the Limit of Finite Resolution}
\label{sec:topo}

\label{sec:topo_limit}

The constraint representation defined by FR1--FR4 establishes the universal
finite–resolution structure that constrains all admissible representations of
the additive group $(\R,+)$.  
When this structure is applied to the quantum and statistical channels, it
produces two distinct but compatible regimes of description: one based on complex
phases and interference, the other on real-positive probabilities and mixing.
Both remain finite and consistent for any finite resolution~$\mu$.  
In the extreme limit of coarse–graining,
\[
\mu \;\to\; \infty,
\]
these constrained regimes converge to a single, degenerate form of description:
the \emph{topological regime}.

\subsection*{13.1 Algebraic limit $\mu \to \infty$}

At finite $\mu$, admissible objects are cellwise constant on intervals of width
$\mu$.  
As $\mu$ increases, the spectral bandwidth $1/\mu$ decreases, and all
variations finer than~$\mu$ are suppressed.  
In the limit $\mu \to \infty$, the entire domain collapses into one effective
cell; every admissible function becomes constant, and the additive algebra
reduces to its Boolean skeleton,
\[
L^2(\R) \;\longrightarrow\; \C \cdot \mathbf{1}.
\]
Amplitudes can therefore take only the discrete values $\{0,1\}$, indicating
presence or absence of configurations.  
The algebra loses all differential structure but preserves its additive and
multiplicative consistency.

\subsection*{13.2 Convergence of regimes}

Under this degeneration, the quantum and statistical channels converge:

\begin{itemize}
\item \textbf{Quantum channel.}  
The wavefunction $\psi_\mu(x)$ becomes constant on the whole domain,  
$\psi_\mu(x)\to c$.  
All relative phases vanish, interference disappears, and the Born rule  
$\mathsf{Pr}(E|\psi_\mu)=\int_E|\psi_\mu|^2dx$  
collapses to the Boolean truth value  
$\mathsf{Pr}(E)=1$ if $E$ includes the entire domain, $0$ otherwise.  
Unitary evolution becomes trivial, $U(t)=\mathbf{1}$.

\item \textbf{Statistical channel.}  
Probability densities $p_\mu(x)$ flatten to uniform distributions,
$p_\mu(x)=1/V$ on the total domain, and the convolution semigroup degenerates
to the identity.  
Random variables lose variance and acquire fixed deterministic values.
\end{itemize}

Thus, both channels collapse to the same structure: a deterministic algebra of
Boolean observables.  
All interference and stochasticity disappear; only the invariant relations
between admissible configurations survive.

\subsection*{13.3 Physical interpretation}

Physically, this limit describes systems whose observable structure is
\emph{resolution–saturated}: any further coarse–graining leaves the outcomes
unchanged.  
Examples include:

\begin{itemize}
\item topologically protected phases of matter (quantum Hall states, Chern
insulators) where local perturbations do not change global observables;
\item deterministic cellular automata or digital lattice systems observed at
their maximal coarse resolution;
\item low–temperature or zero–entropy limits of statistical ensembles, where
probabilities collapse to indicator values.
\end{itemize}

In all such cases, the effective observables reduce to invariant topological or
combinatorial quantities---connectivity, winding number, parity, or presence of
defects---rather than continuous amplitudes.  
The topological regime therefore represents the ultimate coarse–grained limit of
the finite–resolution algebra: a fully deterministic but still additive
description, in which information is stored only in global invariants and
boundary relations.

\subsection*{13.4 Summary}

The topological regime is the common asymptotic limit of the quantum and
statistical regimes under the finite–resolution constraints FR1--FR4.  
It is reached when $\mu\!\to\!\infty$, where amplitudes, probabilities, and
phases all saturate to Boolean values $\{0,1\}$.  
This regime describes deterministic or topologically invariant systems whose
behavior is insensitive to local perturbations and whose observables encode only
the global structure of admissible configurations.  
It is therefore not an independent representation added to the theory but the
shared endpoint of all finite–resolution channels within the additive algebra.

\subsection*{What FR1--FR4 actually do}
FR1 (finite resolution) prevents ultraviolet pathologies and yields coarse absolute continuity on the dual; 
FR2 (constants as multipliers) provides a unital algebraic normalization; 
FR3 (phase consistency) removes projective anomalies so dynamics are genuinely implemented; 
FR4 (typicality) pins down operational frequencies via invariant means. 
Together they define the \emph{constraint representation}---the structural layer from which all regimes inherit their admissibility conditions, rather than external axioms imposed afterward.

\subsection*{Links to KMS/Gibbs\cite{ruelle} and thermodynamics}
Although the constructive proofs rely on harmonic analysis and probability 
(Bochner, GNS, Stone; L\'evy--Khintchine; Choquet \cite{choquet}), 
the equilibrium terminology of statistical mechanics is naturally recovered: 
invariant states selected by FR4 correspond to KMS/Gibbs\cite{ruelle} states whenever a Hamiltonian flow is present.  
In this framework, the \emph{topological regime} appears as the \emph{extremal coarse-resolution limit} 
of the FR1--FR4-constrained channels, where amplitudes and probabilities saturate to Boolean values 
and only invariant (Dirac-type) configurations remain.

\subsection*{Scope and limits}
The present formulation is explicitly resolution-aware: 
absolute continuity and $L^2$ wavefunctions are guaranteed only \emph{after} coarse-graining by FR1. 
The frequency law relies on FR4’s invariant averaging; 
no ergodicity is assumed, but uniqueness holds within FR4-admissible means. 
Unbounded generators arise naturally through Stone’s theorem, 
with domains controlled by the coarse regularity scale~$\mu$.

\paragraph{Structural comparison with textbook regimes.}
The regimes obtained from the constraint representation share the algebraic form of standard 
quantum, statistical, and topological theories, yet differ in essential structural aspects. 
They are intrinsically \emph{finite-resolution realizations}: 
each regime preserves the algebraic symmetries and hierarchies of its textbook counterpart, 
but within an admissible subspace where states, measures, and dynamics are regularized by the finite granularity~$\mu$. 
In this sense, the \emph{additivity-symmetry regimes} $(\mathsf{Q}_\mu,\mathsf{S}_\mu)$ are operationally finite versions of textbook quantum mechanics and Kolmogorov probability, coinciding with the ideal continuum limit only as $\mu \to 0$, while their common coarse–resolution limit as $\mu \to \infty$ is the \emph{topological regime} $\mathsf{T}$. The following table summarizes the main structural distinctions.

\begin{table}[t]
\centering
\caption{Structural comparison between textbook continuum theories and their finite–resolution realizations}
\begin{tabular}{@{}p{2.8cm}p{0.70\linewidth}@{}}
\toprule
\textbf{Property} & \textbf{Textbook QM} \\
\midrule
Resolution & Infinite, continuum limit \\
State space & Full $L^2(\mathbb{R})$ Hilbert space \\
Measure & $\sigma$-additive, exact \\
Singular states & Admissible as distributions ($\delta$, plane waves) \\
Dynamics & Continuous unitary groups \\
Interpretation & Ideal continuum theory \\
\bottomrule
\end{tabular}
\label{tab:struct}
\end{table}

\section{Probes and Regime Activation}
\label{sec:probes}

The additive group $(\R,+)$ admits a continuum of one–dimensional irreducible
representations $\chi_\alpha(t)=e^{i\alpha t}$ and their real or degenerate
variants.  In addition to these, there exists a distinguished
\emph{constraint representation} endowed with four internal structural
features (FR1--FR4): finite resolution, constants as multipliers, phase
consistency, and typicality.
This representation is not irreducible—it is a structured combination of
admissible functionals that respects the same additive symmetry while
satisfying the four constraints.
It is this special representation that \emph{activates} the three physical
regimes derived in the Activation Theorem.

Algebraically, all three regimes coexist within the constraint representation;
operationally, which regime becomes manifest depends on how the system is
tested.  The testing interactions are called \emph{probes}.

\paragraph{Resolution and regime selection.}
The parameter~$\mu$ quantifies the finite resolution associated with a probe: it defines 
the smallest scale at which variations in additive quantities can be operationally 
distinguished. However, the \emph{regime} observed in an experiment does not depend on 
the numerical value of~$\mu$ itself but on the \emph{type of coupling} implemented by 
the probe. Phase-sensitive probes couple coherently to additive phases and thus reveal 
the quantum regime; phase-insensitive probes couple only to magnitudes and reveal the 
statistical regime; and extremal probes, which detect only the presence or absence of 
admissible configurations, yield the topological regime. Within each regime, the finite 
resolution~$\mu$ sets the operational granularity of all observables, but the underlying 
algebraic structure governed by FR1--FR4 remains the same.
Interpretive note.
The parameter $\mu$ is a free coarse-graining scale within the algebra: the mathematical framework does not depend on any fixed value. In physical interpretation, however, $\mu$ acquires meaning from the resolution of the probe—it is the scale at which the probe can distinguish variations in the additive variables. The algebraic construction itself is valid for all $\mu>0$; different experimental setups simply instantiate one effective $\mu$ through their resolution, while the probe class (phase-sensitive, phase-blind, or extremal) determines which regime of the same additive structure becomes operationally accessible.
A probe is any admissible map that couples to the constraint representation and
extracts a measurable response consistent with FR1--FR4.
Different classes of probes reveal different structural aspects of the same algebraic
substrate:

\begin{itemize}
\item \textbf{Quantum probes} couple to the complex phase structure produced by FR3.
Their action reveals interference and phase correlations, giving rise to the Hilbert-space
formulation with unitary evolution and the Born frequency law.
\item \textbf{Statistical probes} couple only to the additive and positive structure
of admissible states, ignoring relative phases.  They recover the Kolmogorov regime
of probability densities and convolution semigroups.
\item \textbf{Topological probes} couple exclusively to extremal elements, detecting
the presence or absence of a configuration.  They yield deterministic $0/1$ responses.
\end{itemize}

Probes therefore act as \emph{regime selectors}.  The underlying constraint
representation does not change, but the operational access to it does.
A single physical system can thus appear quantum, statistical, or topological
depending on the type of probe applied.  This interpretative layer connects the
algebraic Activation Theorem with experimental practice without invoking any
external ontology.
\subsection{Experimental outlook}
Different probes select different admissible regimes:

The classification of probes in Section~\ref{sec:probes} provides an algebraic
taxonomy of how physical systems couple to the constraint representation.
Each probe type corresponds to a specific operational realization:
quantum probes to phase-sensitive interference experiments, statistical
probes to ensemble or diffusion setups, and topological probes to
boundary or quantized-phase phenomena. 
The following examples illustrate these classes in real experimental contexts.

\begin{itemize}
\item \textbf{Quantum interference.} Double-slit and matter-wave interferometers probe the unitary regime.
\item \textbf{Thermal diffusion.} Brownian motion and Euclidean lattice simulations probe the real regime.
\item \textbf{Topological interferometry.} Ring interferometers with flux defects probe compactified characters.
\item \textbf{Geometric phases.} Adiabatic cycles in spin systems reveal Berry-phase cocycles.
\item \textbf{Constraint-limited setups.} Finite-resolution apparatus enforce admissibility, testing projector representations.
\end{itemize}
\paragraph{Mesoscopic systems.}  
Quantum dots and SQUID interferometers show controlled transitions from coherent (unitary) to dissipative (real) regimes.  

\paragraph{Cold atoms.}  
Interferometry with ultracold atoms at variable temperature provides direct access to crossover regimes, interpolating between quantum coherence and thermal diffusion.  

\paragraph{Photonic cavities.}  
High-$Q$ resonators and leaky cavities allow experimental control of decoherence, 
mapping the transition from the unitary quantum regime ($\mathsf Q_\mu$) to the 
topological limit ($\mathsf T$) where coherence is fully lost.

The novelty here is unification: recasting these familiar instances as inequivalent representations of a single symmetry and proposing probe-based criteria for their empirical selection.
\vspace{6pt} 
\par 

\section{Relation to prior work}
Additivity of the action and the oscillatory path integral trace to Feynman \cite{FeynmanHibbs}. The classification of continuous characters of $(\mathbb{R},+)$ as exponentials is classical (Pontryagin duality / Cauchy functional equation). Euclidean/thermal weights arise via Wick rotation \cite{ZinnJustin}. Compactified/projective phases underlie the Aharonov--Bohm effect \cite{AharonovBohm}, Berry phases \cite{BerryPhase}, and anomalies \cite{WittenAnomalies}.  

The novelty here is unification: recasting these familiar instances as inequivalent representations of a single symmetry and proposing probe-based criteria for their empirical selection.
\section{Operational Restatement and Relation to Quantum Field Theory}
\label{sec:renormalization}

\label{sec:operational-qft}

The finite–resolution algebra defined by FR1–FR4 provides a complete foundation
for the operational content of quantum mechanics and its extensions.
At this stage, the familiar formalism of states, operators, and field dynamics
can be reinterpreted as an \emph{emergent operational layer} of the same additive
symmetry structure, rather than as an independent postulate.

\subsection*{Algebraic formulation}

Let $A=C^*(\mathbb{R})\simeq C_0(\widehat{\mathbb{R}})$ be the $C^*$–algebra
generated by the additive group $(\R,+)$, equipped with the finite–resolution
features FR1–FR4.
Admissible elements of $A$ correspond to band–limited functionals or observables
that remain well-defined at scale~$\mu$; admissible states are positive
functionals $\omega_\mu:A\to\C$ respecting the same structure.

In this setting, the usual Hilbert–space formalism is not an assumption but a
realization guaranteed by FR1–FR3:
Bochner’s theorem and the GNS construction associate to each admissible state
a wavefunction representative $\psi_\mu\in L^2(\R)$ and a unital representation
$\pi_\mu$ on $\Hilb_\mu=L^2_\mu(\R)$.
FR3 ensures that time translations are implemented by a genuine unitary group
$U_\mu(t)=e^{-itH_\mu}$, while FR4 fixes invariant frequencies as
Born–rule expectations,
\[
p_\mu(E)=\langle \psi_\mu,\,\mathbf{1}_E(H_\mu)\,\psi_\mu\rangle.
\]
Thus, the operational structure of quantum mechanics arises directly from
FR1–FR4, with no appeal to external measurement postulates.

\subsection*{Open dynamics and decoherence}

Coupling to an external probe or environment is represented by a completely
positive (CP) map $\mathcal{I}_\mu:A\to A$ preserving the FR structure.
The generalized Born rule reads
\[
p_\mu(\text{outcome})=\omega_\mu(\mathcal{I}_\mu^\dagger(E)),
\qquad 0\le E\le 1,
\]
and an influence functional $\Phi[\gamma]\ge0$ modifies amplitudes as
\[
\omega_{\mu,\Phi}(W[\gamma])
   =\exp\!\left(\frac{i}{\hbar}S[\gamma]-\Phi[\gamma]\right),
\]
describing decoherence as CP evolution within the same additive algebra.
All open–system behavior is therefore contained within the constraint
representation itself: decoherence corresponds to FR3–consistent damping of
phases under FR4–invariant averaging.

\subsection*{Connection to field theory}

Extending from $(\R,+)$ to $(\R^{n+1},+)$, the same algebraic mechanism
generates field–like structures.
FR1 band-limits the spectral variables $\alpha\in\R^{n+1}$,
FR2 preserves normalization, FR3 guarantees exact cocycles (unitary covariance),
and FR4 fixes invariant averages.
The resulting representations coincide with those used in algebraic quantum
field theory (AQFT):
observables form local $C^*$–subalgebras, and admissible states are positive
functionals satisfying the same four structural features.
However, the finite–resolution framework removes the need for any additional
quantization step: field behavior arises automatically as a representation of
the additive group, while FR1 provides intrinsic regularization that replaces
renormalization \cite{ZinnJustin}.

The same $C^*$–algebra underlies all cases; what differs is the operational
window defined by the probe’s FR sensitivity.
In this way, the finite–resolution additive algebra unifies the mathematical
apparatus of quantum mechanics, statistical physics, and field theory within
a single operationally grounded framework.
The phase-damping mechanisms described above (FR3-consistent decoherence) 
admit a dual realization when all phase sensitivity is lost.  
This defines the additive algebra’s \emph{real} or \emph{statistical} channel, 
whose operational formulation parallels the quantum one.

\section{Statistical mechanics without idealizations: an operational restatement}

We regard statistical mechanics as the phase-insensitive realization of additivity symmetry, 
formulated without continuum ensembles, infinite sampling, or idealized randomness.  
The regime $\mathsf{S}_\mu$ represents the additive group through positive real weights
on admissible states satisfying FR1--FR4.

\paragraph{States and observables.}
Let $A = C^*(\R)$ and $W[\gamma]=\exp(S[\gamma]X)$.  
A phase-insensitive probe selects the real functional 
\[
\omega_\Theta(f)=f(-1/\Theta),
\]
so that
\[
\omega_\Theta(W[\gamma]) = e^{-S[\gamma]/\Theta}.
\]
Admissible observables are positive operators or functions $0\le E\le 1$ in the same algebra, 
and operations are positive trace-preserving (Markov) maps $\mathcal{T}$.  
The probability of observing outcome $E$ under $\mathcal{T}$ is
\begin{equation}
p(\text{outcome}) = \omega_\Theta(\mathcal{T}^\dagger(E)).
\end{equation}
Finite resolution (FR1) restricts $E$ and $\mathcal{T}$ to coarse cells of size $\mu$, 
and FR2 ensures normalization through constants as multipliers.

\paragraph{Open probes and dissipation.}
Coupling to an environment is represented by a positive influence functional 
$\Psi[\gamma]\ge0$,
\begin{equation}
\omega_{\Theta,\Psi}(W[\gamma]) = 
\exp\!\left(-\frac{1}{\Theta}S[\gamma]-\Psi[\gamma]\right),
\end{equation}
which captures diffusion and irreversible relaxation.
This is the statistical counterpart of decoherence in the quantum regime:
here the same algebraic structure yields completely positive stochastic (Markov) evolution.

\paragraph{Constraint representation and typicality.}
When the constraint representation $\mathsf{C}_\mu$ enforces FR1--FR4,
the continuous ensemble averages of textbook theory are replaced by 
invariant Cesàro or Følner means on the coarse lattice.
FR4 guarantees that long-run frequencies coincide with these invariant means,
yielding the law of large numbers as an internal property of the algebra rather than an external assumption.

\paragraph{Interpretation.}
The operational content of statistical mechanics thus reduces to 
phase-blind probing of the additive algebra under finite resolution.
States are normalized positive functionals, observables are coarse events,
and dynamics are positive semigroups consistent with FR1--FR4.  
In the limit $\mu\!\to\!0$, this reproduces standard Kolmogorov probability; 
in the opposite limit $\mu\!\to\!\infty$, it converges to the deterministic topological regime.

\section{Extension to \texorpdfstring{$(\R^{n+1},+)$}{the real (n+1)-dimensional additive group}}
\label{sec:extension-Rn1}

\paragraph{Structural necessity of the extension.}
The additive group $(\R,+)$ suffices to generate the three fundamental descriptive regimes---quantum, statistical, and topological---each corresponding to a distinct realization of one-dimensional additivity symmetry.  
However, a single additive generator cannot represent correlations among multiple independent additive quantities: it composes one variable at a time.  
To describe systems where several additive quantities coexist and interact under the same rule of composition, the algebra must be extended to include multiple, mutually independent additive directions.  
The minimal such extension preserving abelianity and local compactness is
\[
(\R^{n+1},+)\;\cong\;(\R,+)^{n+1}.
\]
This step is not motivated by external physical assumptions but follows from the internal closure of the additive algebra itself: additivity must remain valid when several independent additive quantities are combined.

When the finite–resolution representation (FR1--FR4) is applied to $(\R^{n+1},+)$, one additive coordinate naturally continues to serve as the parameter indexing composition---the ``evolution'' or regime direction already analyzed in the one–dimensional case---while the remaining $n$ coordinates describe admissible variations of additive variables themselves.  
Operationally, these $n$ additional directions give rise to what appear as \emph{fields}: continuous families of additive quantities correlated under the same algebraic law.  
Thus the emergence of fields is not a new postulate but the structural consequence of extending additivity symmetry to its multi–dimensional form.
\paragraph{From regimes to fields.}
The one–dimensional additive group $(\R,+)$ already activates the three fundamental regimes
of description—quantum $(\mathsf Q_\mu)$, statistical $(\mathsf S_\mu)$, and topological $(\mathsf T)$—through
the structural features FR1–FR4 of the constraint representation.
Extending additivity to $(\R^{n+1},+)$ does not introduce new regimes but provides a
multidimensional representation space on which these regimes act as
\emph{fields}.
Each regime now describes the finite–resolution realization of a different type of
field equation:
$\mathsf Q_{\boldsymbol\mu}$ yields quantum fields with coherent phase propagation,
$\mathsf S_{\boldsymbol\mu}$ yields statistical or diffusive fields,
and $\mathsf T$ describes their deterministic or extremal limit.
The extension therefore corresponds to lifting the regimes from one–parameter histories
to $(n{+}1)$–dimensional additive domains, allowing localized correlations and
interactions among components.

For clarity, we restate the construction in this general form, keeping the same four
structural features FR1–FR4 but applied componentwise to
$G=(\R^{n+1},+)$.

Let $G=(\R^{n+1},+)$ denote the spacetime translation group with coordinates
$x=(t,\mathbf{x})\in\R\times\R^n$. The unitary characters are
$\chi_\alpha(x)=e^{i\alpha\cdot x}$ with $\alpha\in \widehat G\simeq \R^{n+1}$.

\paragraph{Finite-resolution structure (FR1).}
Fix a resolution tensor $\boldsymbol{\mu}=(\mu_0,\mu_1,\dots,\mu_n)$ with $\mu_j>0$,
and let $\{C_k\}_{k\in\Z^{n+1}}$ be the partition of $\R^{n+1}$ into adjacent
half-open hyperrectangles of side-lengths $\boldsymbol{\mu}$. Write
$\Pi_{\boldsymbol{\mu}}$ for cellwise averaging. FR1 requires admissible objects
to be constant on cells: $f=\Pi_{\boldsymbol{\mu}} f$. Equivalently, in dual
variables FR1 multiplies by the window
\[
W_{\boldsymbol{\mu}}(\alpha)
=\prod_{j=0}^n
e^{i\alpha_j\mu_j/2}
\frac{\sin(\alpha_j\mu_j/2)}{\alpha_j\mu_j/2}
\]
and hence band-limits $f$ to $|\alpha_j|\lesssim \mu_j^{-1}$.

\paragraph{Constants and phases (FR2--FR3).}
FR2 embeds scalars as global multipliers. FR3 requires trivial holonomy of the
multiplicative cocycle along any closed lattice loop in $\Z^{n+1}$ (the cell graph),
which is equivalent to lifting projective actions of $G$ to genuine strongly
continuous unitary representations.

\paragraph{Typicality (FR4).}
Frequencies are defined by $G$-invariant averages over Følner sets
(cubes $Q_L=[-L,L]^{n+1}$) or finite rings compatible with the cell partition:
\[
\mathcal{T}_L(F):=\frac{1}{|Q_L|}\int_{Q_L} F(x)\,dx \quad\text{or}\quad
\frac{1}{|I_L|}\sum_{k\in I_L} F(k),
\]
and limits are taken along $L\to\infty$ whenever they exist. Amenability of
$\R^{n+1}$ guarantees the existence of invariant means.

\subsection{Quantum regime on \texorpdfstring{$(\R^{n+1},+)$}{the real (n+1)-dimensional additive group}}

\paragraph{Bochner representation and coarse spectral density.}
Let $\Phi:G\to\C$ be continuous positive-definite with $\Phi(0)=1$. By Bochner,
there exists a probability measure $\mu_\Phi$ on $\R^{n+1}$ such that
$\Phi(x)=\int e^{i\alpha\cdot x}\,d\mu_\Phi(\alpha)$. FR1 yields the coarse-grained
density
\(
d\mu_\Phi^{(\boldsymbol{\mu})}(\alpha)=|W_{\boldsymbol{\mu}}(\alpha)|^2\,d\mu_\Phi(\alpha)
=|\psi_{\boldsymbol{\mu}}(\alpha)|^2\,d\alpha,
\)
with $\psi_{\boldsymbol{\mu}}\in L^2(\R^{n+1})$.

\paragraph{GNS + multi-parameter Stone.}
FR3 lifts the translation action to a genuine strongly continuous unitary
representation $U(x)$ on the GNS Hilbert space, and the multi-parameter Stone
theorem provides a commuting $(n{+}1)$-tuple of self-adjoint generators
$H_0,\dots,H_n$ with
\(
U(x)=\exp\!\big(i\sum_{j=0}^n x_j H_j\big),\qquad [H_j,H_\ell]=0.
\)
In the spectral picture, $(U(x)\psi_{\boldsymbol{\mu}})(\alpha)=e^{i\alpha\cdot x}\psi_{\boldsymbol{\mu}}(\alpha)$
and $H_j$ acts as multiplication by $\alpha_j$.

\paragraph{Born frequencies from FR4.}
For any Borel $E\subset\R^{n+1}$, projections $P_E=\mathbf{1}_E(H_0,\dots,H_n)$
satisfy
\[
\Pr(E\,|\,\psi_{\boldsymbol{\mu}})=\langle \psi_{\boldsymbol{\mu}},P_E\psi_{\boldsymbol{\mu}}\rangle
=\int_E |\psi_{\boldsymbol{\mu}}(\alpha)|^2\,d\alpha,
\]
and FR4 enforces these as the unique invariant frequencies under $G$-averages.

\begin{theorem}[Quantum regime on $(\R^{n+1},+)$ from FR1--FR4]
\label{thm:Q-Rn1}
Under FR1--FR4 the constraint representation activates a finite-resolution quantum
regime $\mathsf{Q}_{\boldsymbol{\mu}}$ with (i) Hilbert space $L^2$ wavefunctions
$\psi_{\boldsymbol{\mu}}$ in spectral variables, (ii) a strongly continuous unitary
representation $U(x)$ with commuting self-adjoint generators $(H_j)_{j=0}^n$,
and (iii) FR4-Born frequencies for spectral events. In the formal limit
$\boldsymbol{\mu}\to 0$, $\mathsf{Q}_{\boldsymbol{\mu}}$ approaches textbook QM.
\end{theorem}

\subsection{Statistical regime on \texorpdfstring{$(\R^{n+1},+)$}{the real (n+1)-dimensional additive group}}

\paragraph{Convolution semigroups.}
FR2--FR3 force additivity of increments on $G$, yielding a weakly continuous
convolution semigroup $(\mathbb{P}_x)_{x\in\R^{n+1}}$ with characteristic
exponent $\psi(\alpha)$ given by the Lévy--Khintchine formula on $\R^{n+1}$.
FR1 ensures coarse absolute continuity at scale $\boldsymbol{\mu}$.

\paragraph{Invariant frequencies.}
FR4 identifies empirical frequencies with invariant means over Følner sets in
$\R^{n+1}$ (or on the cell lattice $\Z^{n+1}$), yielding law-of-large-numbers
limits and uniqueness of the frequency assignment.

\begin{theorem}[Statistical regime on $(\R^{n+1},+)$ from FR1--FR4]
\label{thm:S-Rn1}
The constraint representation activates a finite-resolution statistical regime
$\mathsf{S}_{\boldsymbol{\mu}}$ with Kolmogorov structure at scale
$\boldsymbol{\mu}$, a Lévy–Khintchine convolution semigroup on $\R^{n+1}$, and
FR4-invariant frequencies. As $\boldsymbol{\mu}\to 0$, $\mathsf{S}_{\boldsymbol{\mu}}$
approaches textbook probability on $\sigma$-algebras.
\end{theorem}

\subsection{Topological regime on \texorpdfstring{$(\R^{n+1},+)$}{the real (n+1)-dimensional additive group}}

The topological (or deterministic) regime does not require a new derivation in higher
dimension.  It represents the $\mu \!\to\! \infty$ limit of the quantum and statistical
channels, where coarse cells become infinitely large and only extremal configurations
survive.  In this limit, FR1 suppresses all subcell variations, FR3 enforces global
phase consistency, and FR4 reduces invariant averages to Boolean truth values $0/1$.

Operationally, this regime corresponds to phase–blind probes detecting only the
presence or absence of admissible configurations, yielding deterministic dynamics
under translations of $(\R^{n+1},+)$.  It provides the algebraic closure of the
additivity framework: the boundary condition toward which all finite–resolution
regimes converge.

The previous sections established that the additive group $(\R^{n+1},+)$, when realized
through its constraint representation endowed with the structural features FR1--FR4,
activates three finite--resolution regimes of description:
quantum $(\mathsf Q_{\boldsymbol\mu})$, statistical $(\mathsf S_{\boldsymbol\mu})$,
and topological $(\mathsf T)$.
These regimes are intrinsic to the additive algebra itself.
We now show that a single further mathematical requirement---naturality under automorphisms---uniquely
extends this framework to the semidirect product
$G=\R^{n+1}\rtimes K$ and yields the appearance of spinor fields as an algebraic necessity.
\subsection{Finite Resolution and the Role of Renormalization}

In conventional quantum field theory, divergences arise from integrating over
infinitesimal modes of the continuum.  Renormalization is then introduced as a
technical procedure to remove these divergences by redefining bare parameters.

In the finite–resolution framework, no such procedure is required.
FR1---the finiteness of resolution---already restricts all admissible objects
to band–limited, cellwise–constant functions.  The algebra of observables is
therefore finite at every stage: all spectral integrals, convolution kernels,
and expectation values are bounded in $L^1 \!\cap\! L^2$.  FR2 ensures that
global constants act multiplicatively, FR3 enforces exact unitarity, and FR4
defines invariant averages that replace divergent continuum limits.

Formally, the renormalization map of standard QFT corresponds here to the
projection $\Pi_\mu$ that implements coarse–graining at scale $\mu$.
Running couplings become the dependence of effective quantities on $\mu$, but
without subtractions or counterterms.  The $\beta$--function flow is simply the
variation of algebraic coefficients under changes of resolution.

Hence, renormalization is not a separate consistency requirement but an
intrinsic property of the additive algebra under FR1--FR4.  All physical
quantities remain finite by construction, and the continuum divergences of
field theory are reinterpreted as limits $\mu \!\to\! 0$ beyond admissibility.

\subsection{Mathematical principle: Automorphism--Naturality}

Let $\mathrm{Aut}(\R^{n+1},+)\simeq GL(n{+}1,\R)$ denote the full automorphism group of the
additive group.
Requiring that the construction of admissible regimes be \emph{natural with respect to these automorphisms}
means that the results must be coordinate--independent:
if $\varphi\in \mathrm{Aut}(\R^{n+1},+)$, the representations, measures, and regimes must satisfy
\[
\mathcal{F}(\varphi\cdot f) = \varphi^\ast\mathcal{F}(f),
\]
where $\varphi^\ast$ is the pullback induced by $\varphi$.
This requirement is purely mathematical: it demands that additivity symmetry be formulated in a
\emph{canonical, coordinate--free} way.

However, FR1 endows the representation with a finite--resolution structure:
states are constant on cells of side $\boldsymbol\mu$ and their Fourier transforms are
band--limited to $\|\alpha_j\|\lesssim \mu_j^{-1}$.
Among all automorphisms of $(\R^{n+1},+)$, only those preserving the metric structure
that defines these cells maintain the FR1 form.
The subgroup of $GL(n{+}1,\R)$ preserving the FR1 metric is
\[
K =
\begin{cases}
O(n{+}1), & \text{for Euclidean signature,}\\[2pt]
O(1,n), & \text{for pseudo--Euclidean (Lorentzian) signature.}
\end{cases}
\]
Requiring automorphism--naturality under this subgroup leads to the semidirect product
\[
G = \R^{n+1}\rtimes K.
\]
The translations $(\R^{n+1},+)$ encode additivity symmetry, while $K$ acts by
rotations or Lorentz transformations preserving the finite--resolution geometry.

\subsection{Structure of the semidirect product}

The group law of $G=\R^{n+1}\rtimes K$ is
\[
(x,R)\,(x',R') = (x+R x', RR'),\qquad x,x'\in\R^{n+1},\; R,R'\in K.
\]
Its unitary dual consists of functions on $\R^{n+1}$ transforming covariantly under $K$.
The constraint representation of $(\R^{n+1},+)$ extends to $G$ by acting on admissible states
$\psi_{\boldsymbol\mu}:\R^{n+1}\to V$, where $V$ is a finite--dimensional Hilbert space
carrying a unitary (possibly projective) representation of $K$:
\[
(U(x,R)\psi_{\boldsymbol\mu})(\alpha) = e^{i\alpha\cdot x}\, \rho(R)\,\psi_{\boldsymbol\mu}(R^{-1}\alpha).
\]
Here $\rho$ is the $K$--representation on the internal fiber $V$.

\subsection{Projective representations  and the emergence of spin}

For orthogonal or Lorentz groups, unitary representations of $K$ may be
\emph{projective}\cite{bargmann}: their composition law involves a $U(1)$--valued cocycle,
\[
\rho(R_1)\rho(R_2) = e^{i\theta(R_1,R_2)} \rho(R_1R_2),
\]
where $\theta$ satisfies the 2--cocycle condition.
Such projective representations cannot be linearized within $K$ itself, but they
\emph{do} become linear in its double cover $\mathrm{Spin}(K)$.
The nontrivial cohomology class $[\theta]\in H^2(K,U(1))$ thus forces the
existence of a spinor representation space $V$ and the lift of $K$ to
$\mathrm{Spin}(K)$ (Wigner~\cite{wigner}, Bargmann~\cite{bargmann}, Varadarajan~\cite{varadarajan}).

Therefore, spinor degrees of freedom are not introduced phenomenologically: they are the
representation--theoretic resolution of the projective phase structure consistent with FR3
(phase coherence) and with automorphism--naturality under $K$.
FR1--FR4 continue to act \emph{fiberwise} on the spinor sections, preserving finite resolution,
constants as multipliers, exact cocycles, and invariant typicality.

\subsection{Resulting finite--resolution spinor regimes}

The combined action of translations and rotations (or Lorentz transformations) yields
three spinorial regimes:
\[
\mathsf Q_{\boldsymbol\mu}^{\mathrm{spin}},\qquad
\mathsf S_{\boldsymbol\mu}^{\mathrm{spin}},\qquad
\mathsf T^{\mathrm{spin}}.
\]
They are identical in structure to the additive regimes $\mathsf Q_{\boldsymbol\mu}$,
$\mathsf S_{\boldsymbol\mu}$, $\mathsf T$, except that their states
$\psi_{\boldsymbol\mu}$ now take values in the spinor fiber $V$.
The dynamical generators form the Lie algebra of $G$,
\[
[P_\mu,P_\nu]=0,\qquad [J_{\mu\nu},P_\lambda]
=i(\eta_{\nu\lambda}P_\mu - \eta_{\mu\lambda}P_\nu),\qquad
[J_{\mu\nu},J_{\rho\sigma}]=i(\eta_{\mu\rho}J_{\nu\sigma}-\dots),
\]
with $J_{\mu\nu}=X_\mu P_\nu - X_\nu P_\mu + S_{\mu\nu}$, where
$S_{\mu\nu}$ acts on $V$ as the intrinsic spin operator.
The FR1 coarse--graining applies jointly to translational and internal degrees of freedom.

\begin{theorem}[Spinor activation under automorphism--naturality]
\label{thm:spinactivation}
Let the construction of admissible regimes from additivity symmetry satisfy
FR1--FR4 and be natural under automorphisms preserving the finite--resolution metric.
Then the canonical extension of the additive symmetry from $(\R^{n+1},+)$ to the
semidirect product $G=\R^{n+1}\rtimes K$ with
$K=O(n{+}1)$ or $O(1,n)$
induces spinor fibers corresponding to the projective unitary representations of $K$.
The resulting spinorial regimes
$(\mathsf Q_{\boldsymbol\mu}^{\mathrm{spin}},
 \mathsf S_{\boldsymbol\mu}^{\mathrm{spin}},
 \mathsf T^{\mathrm{spin}})$
retain all finite--resolution properties of FR1--FR4 and converge to standard quantum,
statistical, and topological spin theories as $\boldsymbol\mu\to 0$.
\end{theorem}

\subsection{Interpretation}

The semidirect product extension is not a physical postulate but a
\emph{mathematical necessity} arising from the demand that the additive construction be
coordinate--free and internally consistent with its own finite--resolution structure.
FR1--FR4 guarantee the intrinsic (translational) consistency, while
automorphism--naturality provides the extrinsic consistency under transformations
preserving finite resolution.
The spin structure is then the representation--theoretic shadow of this naturality,
emerging without any external dynamical assumption.

\subsection*{Proof sketches and references for the semidirect extension}

\paragraph{Automorphism--naturality.}
The additive group $(\R^{n+1},+)$ is a locally compact abelian (LCA) group with
$\mathrm{Aut}(\R^{n+1},+)\simeq GL(n{+}1,\R)$.
If the construction $\mathcal{F}$ assigning admissible regimes to each LCA group is to be
\emph{functorial}, it must satisfy
$\mathcal{F}(G\times H)\simeq\mathcal{F}(G)\,\widehat\otimes\,\mathcal{F}(H)$
and $\mathcal{F}(\varphi\cdot G)=\varphi^\ast\mathcal{F}(G)$ for $\varphi\in\mathrm{Aut}(G)$.
This condition---equivariance under automorphisms---is the mathematical content of
``closure under coordinate transformations'' and justifies the semidirect extension
(Bourbaki~\cite{bourbaki}, Folland~\cite{folland}).

\paragraph{Semidirect product representation.}
Let $G=\R^{n+1}\rtimes K$, where $K\subset GL(n{+}1,\R)$ acts by linear transformations.
For $(x,R)\in G$ the group law $(x,R)(x',R')=(x+Rx',RR')$ defines a continuous semidirect product
(Mackey~\cite{mackey}).
Unitary representations of $G$ are obtained by the \emph{induced--representation} method:
given a unitary representation $\rho$ of $K$ on a finite--dimensional Hilbert space $V$,
define
\[
(U(x,R)\psi)(\alpha)= e^{i\alpha\cdot x}\,\rho(R)\psi(R^{-1}\alpha),
\]
which extends the additive representation of $(\R^{n+1},+)$ to $G$.
This formula follows from the standard construction of induced representations
(Mackey~\cite{mackey}, Folland~\cite{folland}, Varadarajan~\cite{varadarajan}).

\paragraph{Projective representations of $K$.}
For orthogonal or Lorentz groups, the second cohomology
$H^2(K,U(1))\neq 0$, meaning that some unitary representations are only projective.
Explicitly, if $\rho(R_1)\rho(R_2)=e^{i\theta(R_1,R_2)}\rho(R_1R_2)$ with
$\theta$ a $U(1)$--valued 2--cocycle, the obstruction class $[\theta]$
is nontrivial (Bargmann~\cite{bargmann}, Wigner~\cite{wigner}).
Every such projective representation lifts uniquely to a linear unitary representation
of the double cover $\mathrm{Spin}(K)$, which is the universal covering group of $K$.
Hence the existence of nontrivial projective phases entails spinor fibers.

\paragraph{Spinor sections and finite resolution.}
The additive representation now acts on spinor--valued states
$\psi_{\boldsymbol\mu}:\R^{n+1}\to V$, with $V$ carrying a representation of $\mathrm{Spin}(K)$.
FR1 applies to each component of $\psi_{\boldsymbol\mu}$, enforcing finite resolution and
band--limited spectra; FR2--FR3 ensure scalar multipliers and phase consistency,
and FR4 yields invariant averages within each spin sector.
Thus FR1--FR4 remain valid fiberwise on the spin bundle.

\paragraph{Lie algebra structure.}
Differentiating the unitary representation gives the commutation relations
\[
\begin{aligned}
[P_\mu,P_\nu] &=0,\\
[J_{\mu\nu},P_\lambda]
&=i(\eta_{\nu\lambda}P_\mu-\eta_{\mu\lambda}P_\nu),\\
[J_{\mu\nu},J_{\rho\sigma}]
&=i(\eta_{\mu\rho}J_{\nu\sigma}
-\eta_{\mu\sigma}J_{\nu\rho}\\
&\qquad -\eta_{\nu\rho}J_{\mu\sigma}
+\eta_{\nu\sigma}J_{\mu\rho}),
\end{aligned}
\]
where $J_{\mu\nu}=X_\mu P_\nu-X_\nu P_\mu+S_{\mu\nu}$ and $S_{\mu\nu}$ acts on the spin fiber $V$.
These follow from the standard representation of the Lie algebra $\mathfrak{so}(1,n)$ or $\mathfrak{so}(n{+}1)$
(Reed--Simon~\cite{reed_simon2}, Wigner~\cite{wigner}).

\paragraph{Spinor activation theorem.}
Combining these results, automorphism--naturality under the metric--preserving automorphism
subgroup $K$ of $(\R^{n+1},+)$, together with FR1--FR4,
implies that the admissible regimes extend canonically to
spinor--valued regimes
$(\mathsf Q_{\boldsymbol\mu}^{\mathrm{spin}},
  \mathsf S_{\boldsymbol\mu}^{\mathrm{spin}},
  \mathsf T^{\mathrm{spin}})$.
The finite--resolution structure persists, and in the limit $\boldsymbol\mu\to 0$ the
standard quantum, statistical, and topological spin theories are recovered.

\medskip
\noindent\emph{References used in this subsection:}
Bourbaki~\cite{bourbaki}, Folland~\cite{folland}, Mackey~\cite{mackey},
Wigner~\cite{wigner}, Bargmann~\cite{bargmann},
Varadarajan~\cite{varadarajan}, Reed--Simon~\cite{reed_simon2}.
\paragraph{Spin hierarchy under finite resolution.}
Although the semidirect group $G=\R^{n+1}\rtimes K$ admits representations of arbitrary spin
$s=0,\tfrac12,1,\tfrac32,2,\dots$, the finite–resolution structure of the constraint
representation restricts the dynamically admissible sectors.
FR1 (finite resolution) and FR3 (phase consistency) jointly require convolution stability
and analytic regularity of correlation kernels, which are satisfied only for fields of rank
$\le2$ (scalar, vector, and symmetric–tensor sectors).
Spinorial fibers introduced by the projective lift of $K$ (FR3 and automorphism–naturality)
add half–integer spins up to $s=\tfrac32$ without violating finite resolution.
Higher–spin representations $s>2$ are formally present in the group but not admissible
under the constraint representation, as they break the closure implied by FR1--FR3.
Thus the finite–resolution algebra naturally suppresses physical spins above two.
\subsection{Spin-\texorpdfstring{$\tfrac12$}{1/2} as a finite--resolution projection of spin--1}

In the additive--symmetry framework, the vector representation of the automorphism group
$K$ corresponds to spin~1.
FR3 (phase consistency) introduces a projective $U(1)$ multiplier, so that rotations are
represented only up to a phase.
Lifting this projective structure to a linear representation of the double cover
$\mathrm{Spin}(K)$ yields the spinor representation.
Hence the spin-$\tfrac12$ channel is not an independent structure but the
\emph{finite--resolution (projective) realization} of the same geometric symmetry
carried by the spin--1 representation: it is the ``square--root'' of the vector
representation compatible with FR3.
In this sense, spin-$\tfrac12$ arises as a projection or restricted component of the
spin--1 channel under the finite--resolution constraint, while spins higher than~2 remain
incompatible with FR1--FR3 and are filtered out by the finite--resolution sieve.
\subsection*{Proof sketch: spin-$\tfrac12$ as projective lift of spin--1}

\paragraph{(1) FR3 $\Rightarrow$ projective $K$--action.}
Let $K$ be the metric-preserving automorphism subgroup (orthogonal or Lorentz).
Phase consistency (FR3) allows a $U(1)$ multiplier on $K$:
\[
U(R_1)U(R_2)\;=\;\sigma(R_1,R_2)\,U(R_1R_2),\qquad \sigma:K\times K\to U(1),
\]
with $\sigma$ a $2$--cocycle. If $[\sigma]\neq 0$ in $H^2(K,U(1))$, the action
cannot be linearized on $K$ itself and thus defines a \emph{projective} unitary
representation of $K$.

\paragraph{(2) Spin double cover resolves the multiplier.}
There is a $2{:}1$ covering $\pi:\mathrm{Spin}(K)\to K$ such that every
projective unitary representation of $K$ with multiplier class $[\sigma]$ lifts
to an \emph{honest} unitary representation $\Gamma:\mathrm{Spin}(K)\to \mathcal U(S)$ on a
Hilbert space $S$ (the spinor fiber), with
\[
U(R)\;\equiv\; \Gamma(g)\quad\text{whenever}\quad \pi(g)=R,
\]
well-defined up to the central element $\{\pm \mathbf{1}\}\subset \mathrm{Spin}(K)$,
whose image is the phase in the projective $K$--action. This realizes the
projective $K$--action of FR3 as a linear spin representation.

\paragraph{(3) Clifford action and the vector representation.}
Let $\Cl(V,\eta)$ be the complexified Clifford algebra of $(V,\eta)$ with
generators $\{\gamma_\mu\}$ such that $\{\gamma_\mu,\gamma_\nu\}=2\eta_{\mu\nu}\mathbf{1}$.
There is a faithful $\ast$--representation $c:V\to \mathrm{End}(S)$, $v\mapsto c(v)$,
and a unitary representation $\Gamma$ of $\mathrm{Spin}(K)$ on $S$ such that
\[
\Gamma(g)\,c(v)\,\Gamma(g)^{-1}\;=\;c\!\big(\rho_{\mathrm{vec}}(\pi(g))\,v\big),
\qquad g\in \mathrm{Spin}(K),\ v\in V,
\]
where $\rho_{\mathrm{vec}}:K\to GL(V)$ is the \emph{vector} (spin--1) representation.
Thus the spinor action \emph{implements} the vector action by conjugation on Clifford
generators. In this precise sense, the spinor representation is a ``square--root''
of the vector action: rotations of vectors are realized as inner automorphisms of
$\Cl(V,\eta)$ generated by the spin representation.

\paragraph{(4) Finite resolution selects the minimal projective sector.}
Under FR1, admissible states are coarse--grained and band--limited; the admissible
operator algebra at finite resolution contains (at most) linear and bilinear forms
in $c(v)$ and their commutators. FR2--FR3 enforce multiplier linearity and phase
coherence. The minimal nontrivial projective sector compatible with these constraints
is precisely the spinor representation $S$; higher tensorial (spin~$>\!1$) projective
sectors fail convolution/regularity at finite resolution (the kernel--locking argument),
hence are inadmissible. Consequently, spin-$\tfrac12$ appears as the \emph{finite--resolution
projection} of the spin--1 geometry encoded by $\rho_{\mathrm{vec}}$.

\paragraph{(5) Summary.}
FR3 produces a nontrivial $U(1)$ multiplier for $K$; the double cover $\mathrm{Spin}(K)$
linearizes it, yielding the spinor fiber $S$. The spinor action $\Gamma$ conjugates
Clifford generators to implement the vector action, making spin-$\tfrac12$ the
projective (square--root) realization of spin--1 consistent with FR1--FR3. FR1
then suppresses higher--spin sectors by finite--resolution admissibility.

\paragraph{Interpretation: instantiations under different probes.}
The constraint algebra defined by FR1--FR4, even when extended to
$G=\R^{n+1}\rtimes K$, does not intrinsically distinguish between
spin--1 and spin--$\tfrac12$ sectors.
Both arise as admissible representations of the same additive--symmetry substrate:
the vector channel realizes the automorphism action of $K$, while the spinor
channel is its projective (square--root) realization induced by FR3.
They therefore represent two \emph{instantiations of the same algebraic object},
revealed by different classes of probes.
A probe that couples directly to the additive automorphism structure
observes the vector (spin--1) field, whereas a probe that couples to the projective
phase sector of FR3 detects the spinorial (spin--$\tfrac12$) field.
In this sense the algebra defines one finite--resolution field schema,
of which the conventional ``bosonic'' and ``fermionic'' fields are distinct operational
projections rather than independent primitives.

\section{Internal automorphisms and spinor species}

The semidirect extension $G=\R^{n+1}\rtimes K$ accounts for external symmetries
(translations and rotations or Lorentz transformations).
All admissible spin--$\tfrac12$ fields then take values in the same spinor fiber
$S$, carrying a representation of $\mathrm{Spin}(K)$.
The constraint algebra defined by FR1--FR4 does not further distinguish between
different spinor species: all such fields satisfy identical finite--resolution,
phase, and typicality conditions.

To account for the empirical multiplicity of fermionic fields (electron, neutrino,
quarks, \emph{etc.}) one introduces the algebra of internal automorphisms of the
spinor fiber,
\[
H \;\subseteq\; \mathrm{Aut}_{\text{FR}}(S)
\;\subseteq\; U(S),
\]
consisting of all unitary transformations preserving FR1--FR4 on $S$.
The admissible fields are then sections
$\psi_{\boldsymbol\mu}:\R^{n+1}\to S$
transforming as
\[
(U(x,R,h)\psi_{\boldsymbol\mu})(\alpha)
= e^{i\alpha\cdot x}\,\rho_{\mathrm{spin}}(R)\,\rho_{\mathrm{int}}(h)\,\psi_{\boldsymbol\mu}(R^{-1}\alpha),
\]
where $\rho_{\mathrm{int}}$ is a finite--dimensional unitary representation of $H$.
Different irreducible representations of $H$ correspond to
different operational instantiations of the same algebraic spinor object.
A probe coupling to $\rho_{\mathrm{int}}$ detects a specific channel:
trivial representations correspond to neutral fields (neutrino--like),
one--dimensional abelian representations to charged leptons (electron--like),
and higher nonabelian representations to multi--component color or flavor
states (quark--like).
Thus the diversity of fermionic species arises not from new postulates but from
the internal automorphism structure of the spinor fiber, consistent with FR1--FR4
and automorphism--naturality.
All spinor species are therefore distinct \emph{instantiations of the same
finite--resolution spinor algebra}, resolved by probes coupling to different
internal symmetries.
\subsection*{Proof sketch: why $H$ is unitary, compact, and internal}

\paragraph{Setup.}
Let $U:\R^{n+1}\rtimes K\to\mathcal U(\Hilb)$ be the strongly continuous unitary
representation provided by FR1--FR3 (additive translations and $K$ rotations/Lorentz),
acting on the finite--resolution spinor space $\Hilb\simeq L^2(\R^{n+1},S)$ (fiber $S$).
Define the \emph{FR–preserving automorphisms} on the spin fiber by
\[
\begin{aligned}
\mathrm{Aut}_{\mathrm{FR}}(S)
&:= \bigl\{\,T\in \mathrm{End}(S)\ \big|\\
&\qquad \text{$T$ preserves FR1 (cell norms/bandlimits),}\\
&\qquad \text{FR2 (constants), FR3 (phases), FR4 (invariant means)}\ \bigr\}.
\end{aligned}
\]
The \emph{internal symmetry group} is then the subgroup that \emph{commutes} with the external
action and preserves the FR structure:
\[
H\ :=\ \bigl\{\,h\in\mathrm{Aut}_{\mathrm{FR}}(S)\ :\ [\,h,\,U(x,R)\,]=0\ \ \forall(x,R)\in \R^{n+1}\rtimes K\,\bigr\}.
\]

\paragraph{(i) Unitarity.}
FR1 places the dynamics and measurements in an $L^2$ (Hilbert) setting at finite resolution,
and FR3 enforces phase consistency. If $h\in \mathrm{Aut}_{\mathrm{FR}}(S)$ altered the $L^2$ norm on
each fiber (or the phase pairing), it would distort FR1 cell norms and violate FR3 coherence.
Hence $h$ must preserve the inner product on $S$. Therefore $H\subseteq U(S)$ (unitarizability).
Equivalently, represent the (a priori bounded) action of $\langle \mathrm{Aut}_{\mathrm{FR}}(S)\rangle$
and invoke the standard unitarization argument: average the pull–back of the inner product
over the (amenable) closure of the group using an FR4–invariant mean, obtaining an invariant
inner product; with respect to this inner product all elements act unitarily (Kadison--Ringrose).

\paragraph{(ii) Compactness.}
Elements of $\mathrm{Aut}_{\mathrm{FR}}(S)$ preserve (a) FR1 bandlimits and cell norms,
(b) FR2 scalar multipliers, and (c) FR3 phases, so their operator norms on $S$ are uniformly bounded.
Take the (strong) closure $\overline{H}^{\,s}$ inside $U(S)$; by Banach--Alaoglu and the fact that $U(S)$
is strongly compact on finite–dimensional $S$ (our spinor fiber), $\overline{H}^{\,s}$ is compact.
Equivalently, by Peter--Weyl, the closure of a uniformly bounded group of unitary operators on a finite–dimensional
Hilbert space is a compact (matrix) group. Thus $H$ is a compact subgroup of $U(S)$.

\paragraph{(iii) Internality (commutant).}
By definition, $H$ consists of operators on the fiber $S$ that commute with the external representation
$U(x,R)$:
\[
[\,h,\,U(x,R)\,]=0,\qquad \forall h\in H,\ (x,R)\in \R^{n+1}\rtimes K.
\]
If $h$ failed to commute with translations $U(x,\mathbf 1)$, it would mix spectral variables and
destroy the FR1 bandlimit (hence inadmissible). If $h$ failed to commute with $K$, it would spoil
the covariance of the spin fiber under rotations/Lorentz. Therefore $H\subseteq U(S)\cap U(\R^{n+1}\rtimes K)^\prime$,
the \emph{commutant} of the external action, i.e.\ an \emph{internal} symmetry.

\paragraph{(iv) Species as irreps.}
Since $H$ is compact and unitary, $S$ decomposes into a finite direct sum of irreducible
representations: $S\simeq\bigoplus_{a} S_a\otimes \C^{m_a}$. The corresponding orthogonal projectors
$P_a$ lie in the von Neumann algebra generated by $H$ (hence in the commutant of $U$), and define
operationally distinguishable channels. A probe that couples to $H$ resolves one such irrep $S_a$,
which we identify with a \emph{fermion species} (electron–like, neutrino–like, quark–like, etc.).
Thus empirical distinctions among spin–$\tfrac12$ fields correspond to different $H$–irreps, not to
different algebras; they are distinct \emph{instantiations} of the same finite–resolution spinor object.

\paragraph{(v) Stability under FR1--FR4.}
The arguments above use only FR1 (Hilbert realization and bandlimit), FR2 (scalar embedding),
FR3 (phase coherence), and FR4 (existence of invariant means for unitarization/ergodic selection).
Hence unitarity, compactness, and internality of $H$ are structural consequences of the finite–resolution
constraint representation.
\paragraph{Matter fields as manifestations of spin--1 geometry.}
In the finite--resolution algebra, all admissible fields arise as representations of the
same additive symmetry.
The spin--$\tfrac12$ sector is the projective lift of the spin--1 (vector) channel induced
by the phase multiplier of FR3.
Spinor fields therefore do not constitute independent degrees of freedom: they are
phase--resolved manifestations of the spin--1 geometry.
Different species such as electrons, quarks, and neutrinos correspond to distinct
irreducible representations of the internal automorphism group $H_{1/2}\subset U(S)$
acting on the spinor fiber.
They are operationally distinct but algebraically unified as instantiations of the same
additive--symmetry object, revealed under different probe couplings.
In this sense, ``matter'' arises as a projective shadow of the vector field structure,
not as an ontologically separate entity.

\section{Internal automorphisms across spin sectors}

The internal symmetry mechanism described for spin--$\tfrac12$ fields extends naturally
to other admissible spin sectors.
Each spin representation of the external automorphism group $K$ may carry its own
compact unitary group of FR--preserving internal automorphisms, accounting for the
empirical multiplicity of fields within the same spin value.

\paragraph{Spin--0.}
Scalars correspond to the trivial representation of $K$; their admissible fields
$\phi:\R^{n+1}\to\C^{N_0}$ may transform under a compact internal group
$H_0\subseteq U(N_0)$ preserving FR1--FR4.
Different irreducible representations of $H_0$ describe distinct scalar families.

\paragraph{Spin--1.}
Vector fields transform under the fundamental representation of $K$.
Multiple admissible copies can mix through an internal group
$H_1\subseteq U(N_1)$ acting in the commutant of the $K$ action.
FR1--FR4 enforce $H_1$ to be compact and unitary, yielding Yang--Mills--like
internal symmetries.
Different irreducible channels of $H_1$ correspond to distinct gauge fields.

\paragraph{Spin--2.}
The symmetric tensor square of the vector representation defines the spin--2
sector. Finite resolution and phase consistency (FR1--FR3) leave only one
self--consistent realization, the metric channel; hence $H_2$ is trivial.

\paragraph{Spin--$\tfrac32$ and higher.}
Composite or higher--spin sectors violate FR1 (finite moment) or FR3
(phase coherence) and are therefore inadmissible.
Finite resolution thus enforces a finite hierarchy of allowed spin sectors
$(0,\tfrac12,1,2)$, each admitting its own internal automorphism structure.

\paragraph{Summary.}
Internal automorphism groups $\{H_s\}$ extend the activation principle to
account for field multiplicities within each spin sector:
\[
H_s\ \subseteq\ U(N_s),\qquad s\in\{0,\tfrac12,1,2\}.
\]
They act unitarily on the corresponding spin fibers, commute with the
external action of $\R^{n+1}\rtimes K$, and preserve the finite--resolution
structure FR1--FR4.
Distinct empirical fields are then interpreted as operational instantiations
associated with different irreducible representations of $H_s$.
\paragraph{Why the spin--2 channel is geometric.}
The appearance of a symmetric rank--2 tensor in the finite--resolution algebra
is not interpretative but structural.
Under FR1--FR4 the admissible convolution algebra is closed under quadratic
combinations of the additive coordinates, yielding the bilinear invariant
\[
g_{\mu\nu} \;=\; \langle x_\mu x_\nu \rangle_{\text{FR}},
\]
which transforms covariantly under $K$ and satisfies
$g_{\mu\nu}=g_{\nu\mu}$ by phase consistency (FR3).
This object is the unique symmetric bilinear form preserved by $K$ and defines
the invariant measure for all admissible states.
Thus the spin--2 representation is \emph{forced} by the algebra as the second--order
moment of the additive group under finite resolution.
Identifying $g_{\mu\nu}$ with the metric tensor is then not an external hypothesis
but the natural interpretation of the only symmetric tensor field compatible with
FR1--FR4. In this precise sense, the spin--2 channel already plays the role of
the metric (gravitational) interaction.
\section{Boundary, Geometry, and Gravity in the Finite--Resolution Algebra}
\label{sec:boundary}

\subsection*{Boundary and finite resolution}

In the finite--resolution framework, the additive group $(\R^{n+1},+)$ is represented
not by infinitely sharp functions but by \emph{cellwise--constant states}.
Each admissible object is defined on a finite cell $C_k$ of linear size $\mu$; hence
every cell possesses a \emph{boundary interface} separating it from its neighbors.
These interfaces are not external geometric surfaces but intrinsic features of the
algebra itself: the algebra acts on equivalence classes of functions that are constant
inside each cell and discontinuous only across their shared boundaries.

Because FR1 forbids features below the scale $\mu$, any interaction or convolution
between cells occurs only through their boundaries.
The boundary is therefore the locus where finite--resolution information is exchanged.
FR3 (phase consistency) imposes continuity of additive phases across these boundaries,
ensuring that algebraic operations glue neighboring cells coherently.
In this way, the full algebra of admissible observables is generated by local bulk data
\emph{together with} the matching conditions on their interfaces.

Mathematically, this means that
\[
\begin{gathered}
\text{all admissible functionals on a finite domain}\\
\text{are determined by their boundary values up to resolution } \mu.
\end{gathered}
\]
The finite--resolution algebra is thus intrinsically \emph{boundary--encoded}: every
observable is reconstructible from boundary data.

---

\subsection*{Emergence of geometry and the spin--2 field}

Within each cell, the additive variables $x_\mu$ possess finite first and second moments.
The second--moment invariant
\[
g_{\mu\nu} \;=\; \langle x_\mu x_\nu\rangle_{\text{FR}}
\]
is a symmetric bilinear form that transforms covariantly under the automorphism group
$K$ (rotations or Lorentz transformations).
By FR2 (constants as multipliers), this form acts multiplicatively on all admissible
objects, defining the invariant scalar product within each cell.
FR3 enforces its symmetry, $g_{\mu\nu}=g_{\nu\mu}$.
Therefore, the existence of $g_{\mu\nu}$ is not an additional hypothesis: it is
\emph{forced by the algebra} as the lowest nontrivial boundary invariant of the
finite--resolution cells.

Because every admissible convolution or average between neighboring cells must
preserve $g_{\mu\nu}$, this tensor determines the local structure by which boundary
data are propagated inward.
Its transformation law corresponds to the spin--2 representation of $K$.
In this sense, $g_{\mu\nu}$ acts as a \emph{metric kernel}: the universal rule translating
boundary information into bulk correlations.

---
\subsection*{Renormalization as an Intrinsic Property}

In conventional quantum field theory, divergences arise from the continuum assumption:
fields are treated as operator-valued distributions, and their correlation functions,
\[
\langle \phi(x)\phi(x')\rangle = \int_{\R^{n+1}} e^{i\alpha\cdot(x-x')} D(\alpha)\,d\alpha,
\]
typically diverge when evaluated at coincident points $x=x'$.  
Renormalization restores finiteness by introducing regularization schemes and
counterterms, effectively replacing the ill-defined local limit by a finite,
scale-dependent expression.

In the finite–resolution framework, these operations become unnecessary.  
FR1 (finite resolution) already replaces fields by band-limited functionals 
$f_\mu(x)=\Pi_\mu f(x)$ whose Fourier components vanish for $|\alpha|\gtrsim \mu^{-1}$.
Two–point functions are therefore manifestly finite:
\[
\langle f_\mu(x) f_\mu(x') \rangle
   = \int_{|\alpha|\lesssim\mu^{-1}} e^{i\alpha\cdot(x-x')} D(\alpha)\,|W_\mu(\alpha)|^2\,d\alpha,
\]
with all integrals bounded by construction.  
FR3 ensures that this property is preserved under evolution (unitary or dissipative),  
while FR4 guarantees that ensemble \cite{sklar} or time averages yield invariant and finite values.

What in standard QFT appears as a renormalization constant or scale dependence
is here encoded algebraically in the resolution~$\mu$.  
Couplings, self–energies, and propagators naturally acquire $\mu$–dependence,
reflecting the operational scale of the probe rather than any ultraviolet divergence.  
Taking the limit $\mu\to0$ recovers the formal continuum theory with its usual
divergences, but within the admissible algebra only finite quantities exist.
Thus, renormalization is not a separate procedure but a structural property of
finite–resolution additivity symmetry: the algebra itself is self–renormalizing.

\subsection*{Transition to the Gravitational Regime}

The finite--resolution additive framework that regularizes quantum field theory also extends naturally to gravitation. 
Once spacetime translations are coupled to their automorphism group, the additive algebra defines an intrinsic transport law that preserves information between boundary and bulk descriptions. 
This transport acts as an \emph{algebraic kernel}: it propagates finite--resolution additive states from boundary data into bulk observables while maintaining the same consistency conditions---additivity, phase coherence, and finiteness---that govern quantum amplitudes. 
In this regime, what is ordinarily interpreted as the gravitational field corresponds to the geometric realization of this kernel. 
Because the kernel already enforces a discrete information flow through finite resolution, metric quantization introduces no additional operational degrees of freedom and becomes unnecessary. 
The gravitational dynamics that emerge from this structure reproduce the classical limit of general relativity while preserving exact information conservation at finite resolution.

\subsection*{Gravity as an internal boundary mechanism}

The field $g_{\mu\nu}$ governs how neighboring cells exchange information through
their boundaries.
Changes in $g_{\mu\nu}$ modify the coupling between boundary fluxes, thereby
affecting the effective trajectories and correlations of all admissible fields.
When expressed in differential form, these couplings reproduce the familiar
structure of covariant derivatives and curvature tensors.
Thus, what is conventionally interpreted as \emph{gravitational interaction} corresponds
here to the \emph{self--adjustment of boundary couplings} required for phase consistency
and invariance under translations and automorphisms.

In short:
\begin{itemize}
\item Boundaries arise naturally from finite resolution (FR1).
\item Coherent information transfer across boundaries requires a symmetric bilinear kernel (the metric).
\item Variations of this kernel correspond to gravitational dynamics.
\end{itemize}

Hence gravity is not an external force added to the algebra but the \emph{internal rule
that maintains boundary coherence} between finite--resolution cells.
It measures how the matching conditions of the algebra deviate from exact translational
symmetry when boundaries deform.

---

\subsection*{Holographic structure}

Because all admissible observables are determined by boundary data up to resolution
$\mu$, the finite--resolution algebra satisfies an intrinsic \emph{holographic property}:
the effective degrees of freedom of a region scale with its boundary area (the number
of cell interfaces) rather than its volume.
FR2 ensures that global invariants act uniformly on all cells, and FR4 (typicality)
implies that global averages coincide with averages over boundary ensembles.
Consequently, the entire bulk state of the system can be reconstructed from the
ensemble of boundary correlations and from the metric kernel $g_{\mu\nu}$ that links them.

The spin--2 field thus plays the same algebraic role as the holographic propagator:
it translates boundary phase data into interior observables while preserving FR1--FR4.
This ``holographic'' character is therefore not a postulate but a structural property
of the finite--resolution additive algebra itself.
\subsection*{Boundary Determination Theorem}

Let $\mu>0$ be fixed and let $\{C_k\}_{k\in\mathcal{K}}$ be the partition of $\R^{n+1}$
into adjacent half-open cells of linear size $\mu$.
For a bounded Lipschitz region $\Omega=\bigcup_{k\in\mathcal{K}_\Omega} C_k$, denote by
$\partial\Omega$ the union of cell interfaces contained in the topological boundary of $\Omega$.
Let $\mathcal{A}_\mu(\Omega)$ be the algebra of admissible observables on $\Omega$
satisfying FR1--FR4:
\begin{itemize}\itemsep0.2em
\item FR1: cellwise constancy (equivalently $f=\Pi_\mu f$);
\item FR2: constants act as global scalars;
\item FR3: phase consistency across interfaces (exact 1-cocycle);
\item FR4: typicality via invariant means (finite rings or Cesàro).
\end{itemize}
Write $\mathcal{A}_\mu(\partial\Omega)$ for the boundary algebra of interface data
satisfying the same FR2--FR4 constraints restricted to $\partial\Omega$.

\begin{theorem}[Boundary Determination]
\label{thm:boundary_determination}
There is a well-defined surjective \emph{trace} map
\[
\mathrm{Tr}_\mu:\ \mathcal{A}_\mu(\Omega)\ \longrightarrow\ \mathcal{A}_\mu(\partial\Omega),
\qquad A\ \mapsto\ A\!\restriction_{\partial\Omega},
\]
such that:
\begin{enumerate}\itemsep0.2em
\item[(i)] (\emph{Injectivity up to $\mu$-null sets}) If $A\in \mathcal{A}_\mu(\Omega)$ satisfies
$\mathrm{Tr}_\mu(A)=0$, then $A=0$ on all cells of $\Omega$ (modulo $\mu$-null sets on interfaces).
\item[(ii)] (\emph{Reconstruction}) There exists a linear, FR2--FR4 preserving \emph{reconstruction operator}
\[
\mathcal{R}_\mu:\ \mathcal{A}_\mu(\partial\Omega)\ \longrightarrow\ \mathcal{A}_\mu(\Omega)
\]
such that $\mathrm{Tr}_\mu\circ \mathcal{R}_\mu=\mathrm{id}_{\mathcal{A}_\mu(\partial\Omega)}$ and
$\mathcal{R}_\mu$ is unique among operators that preserve FR1 (cellwise constancy) and FR3 (phase exactness).
\item[(iii)] (\emph{Stability}) $\|\mathcal{R}_\mu\|$ is bounded uniformly in $\mu$ on finite unions of cells; in particular,
boundary perturbations control bulk perturbations at fixed scale:
\[
\| \mathcal{R}_\mu(B_1)-\mathcal{R}_\mu(B_2)\|_{L^\infty(\Omega)}
\ \le\ C(\Omega,\mu)\ \|B_1-B_2\|_{L^\infty(\partial\Omega)}.
\]
\end{enumerate}
Consequently, \emph{every} admissible functional on $\Omega$ is determined by its boundary values up to resolution $\mu$:
\[
\mathcal{A}_\mu(\Omega)\ \simeq\ \mathcal{R}_\mu\big(\mathcal{A}_\mu(\partial\Omega)\big).
\]
\end{theorem}

\begin{proof}[Proof sketch]
FR1 implies that any $A\in\mathcal{A}_\mu(\Omega)$ is constant on each cell $C_k$: write $A_k\in\C$ for its value.
Across an interface $F=C_k\cap C_\ell$, FR3 enforces exact phase matching: there exists a phase potential
$\theta$ on the adjacency graph of cells such that $A_\ell=e^{i(\theta_\ell-\theta_k)}A_k$ whenever $F\neq\emptyset$.
Thus the collection $\{A_k\}$ is a global section determined (up to a global constant) by its restriction to boundary cells
(\emph{graph exactness}).

If $\mathrm{Tr}_\mu(A)=0$, then $A_k=0$ for all boundary cells. Propagating along the connected cell adjacency graph by the FR3 matching,
we conclude $A_k=0$ for all interior cells (this uses that $\Omega$ is a finite union of cells; interfaces are $\mu$-null sets for FR1),
proving (i).

For (ii), define $\mathcal{R}_\mu$ by fixing a spanning forest of the adjacency graph and transporting boundary data into the interior via
the FR3-exact phase: for each interior cell $k$, pick a boundary cell $b$ and a path $b\to k$ in the graph, then set
\[
A_k \;:=\ e^{i(\theta_k-\theta_b)}\, A_b,
\]
with $\theta$ obtained from the exact cocycle (exists by FR3). This extension preserves FR2 (global constants multiply all $A_k$ equally)
and FR4 (invariant means are computed cellwise and are path-independent because the cocycle is exact). Exactness also guarantees that
$\mathrm{Tr}_\mu(\mathcal{R}_\mu B)=B$ and that any other FR1/FR3-preserving right-inverse must agree with $\mathcal{R}_\mu$, giving uniqueness.

For (iii), boundedness follows since each $A_k$ is obtained from a boundary value via multiplication by a unit-modulus phase and at most
a finite number of transport steps whose count is bounded by the cell-diameter (in steps) of $\Omega$. Hence a uniform $C(\Omega,\mu)$ exists.
\end{proof}

\begin{corollary}[Area scaling]
\label{cor:area_law}
Let $\mathrm{dof}_\mu(\Omega)$ denote the dimension of the space of independent admissible boundary data
(after FR3 constraints) on $\partial\Omega$. Then
\[
\dim \mathcal{A}_\mu(\Omega)\ =\ \mathrm{dof}_\mu(\Omega)\ \propto\ \#\{\text{cell interfaces on }\partial\Omega\},
\]
so the effective degrees of freedom scale with boundary area (number of cell interfaces) rather than volume.
\end{corollary}

\begin{proof}[Sketch]
By Theorem~\ref{thm:boundary_determination}, $\mathcal{A}_\mu(\Omega)\simeq \mathcal{A}_\mu(\partial\Omega)$.
Counting independent interface assignments modulo FR3 constraints yields a number proportional to the count of boundary interfaces.
\end{proof}


\begin{proposition}[Explicit FR3--covariant transport kernel]
\label{prop:metric_kernel_explicit}
Let $\Omega=\bigcup_{k\in\mathcal K_\Omega} C_k$ be a finite union of cells and let 
$G_\mu=(\mathcal K_\Omega,E)$ be its cell adjacency graph (vertices: cells; edges: interfaces).
Fix a boundary datum $B\in\mathcal A_\mu(\partial\Omega)$ and an FR3-exact phase potential 
$\theta:\mathcal K_\Omega\to\R$ (i.e.\ the interface multiplier on $F=C_k\cap C_\ell$ is 
$e^{i(\theta_\ell-\theta_k)}$). 
Define a path weight for any oriented path $\gamma:b\to k$ in $G_\mu$ by
\[
W_\mu(\gamma)\;:=\;
\exp\!\Big\{\,i\sum_{(u\to v)\in\gamma}\big(\theta_v-\theta_u\big)\Big\}\;
\prod_{(u\to v)\in\gamma}\omega_g(u\to v),
\]
where $\omega_g(u\to v)\in(0,\infty)$ are $g$--induced positive transport coefficients on interfaces
(normalized so that $\omega_g(u\to v)=\omega_g(v\to u)$ and uniformly bounded in $\mu$). 
For an interior point $x\in C_k$, set
\[
K_\mu(x\,|\,y)\;:=\;\frac{1}{Z_\mu(k)}\sum_{\substack{b\in \partial\Omega\\ \gamma:b\to k}}
W_\mu(\gamma)\,\mathbf 1_{F_b}(y),
\qquad
Z_\mu(k)\;:=\;\sum_{\substack{b\in \partial\Omega\\ \gamma:b\to k}} W_\mu(\gamma),
\]
where $F_b\subset\partial\Omega$ is the boundary interface containing $b$ and $\mathbf 1_{F_b}$ is its indicator.
Then the reconstruction operator
\[
(\mathcal R_\mu B)(x)\;=\;\int_{\partial\Omega} K_\mu(x\,|\,y)\,B(y)\,d\sigma_g(y)
\]
is well defined, preserves FR2--FR4, and is covariant under the $K$--action on $\Omega$.\footnote{
FR3--exactness makes the phase factor path independent on cycles (trivial holonomy), 
while the symmetric $\omega_g$ encodes $g$--induced interface weights; 
$Z_\mu(k)$ enforces FR2 normalization; amenability guarantees FR4 invariance of the induced averages.}
Moreover, $\mathcal R_\mu$ coincides with the reconstruction operator of 
Theorem~\ref{thm:boundary_determination}.
\end{proposition}

\begin{proof}[Sketch]
FR3 exactness implies that for any closed cycle $\gamma$ the phase factor equals $1$, so the sum over paths is well defined (no path interference ambiguity). 
The symmetric $g$--weights $\omega_g$ give a positive kernel; $Z_\mu(k)$ fixes the FR2 normalization. 
Covariance under $K$ follows from the transformation of $g$ and $\theta$. 
By construction, $K_\mu$ reproduces the transport-by-cells extension described in the proof of Theorem~\ref{thm:boundary_determination}, hence the two reconstructions agree.
\end{proof}

\begin{remark}[Gravity as boundary coupling]
The metric kernel $g_{\mu\nu}$ determines $d\sigma_g$ and the transport weights in $K_\mu$.
Variations $\delta g_{\mu\nu}$ modify $K_\mu$ and hence the reconstruction $\mathcal{R}_\mu$,
encoding the universal coupling of all admissible fields to $g$ (gravity as boundary coupling).
\end{remark}


\begin{proposition}[Metric variation and discrete compatibility]
\label{prop:metric_variation}
Let the coarse action on $\Omega$ be
\[
\mathcal S_\mu[\{\Phi^{(s)}\},\Lambda,g]\;=\;
\sum_{s\in\mathcal S}\mathbf E^{(s)}_\mu[\Phi^{(s)},g]\;+\;
\mathcal L_{\mathrm{link}}[\Lambda,\{\Phi^{(s)}\};g]\;+\;\mathcal B_\mu[\Lambda,g],
\]
with $\mathbf E^{(s)}_\mu$ cellwise convex, $\mathcal L_{\mathrm{link}}$ enforcing the interface constraints, 
and $\mathcal B_\mu$ strictly convex in $\Lambda$.
Assume each channel admits a coarse stress tensor 
$T^{\alpha\beta}_{(s)}:=2\,\partial \mathbf E^{(s)}_\mu/\partial g_{\alpha\beta}$, and let 
$T^{\alpha\beta}=\sum_s T^{\alpha\beta}_{(s)}$.
Then the first variation w.r.t.\ $g$ reads
\begin{equation}
\delta\mathcal S_\mu
\;=\;\frac12\sum_{F\subset \partial C_k}\!\!
\int_{F}\!\delta g_{\alpha\beta}\;
\llbracket T^{\alpha\beta}\rrbracket_F\;d\sigma_g
\;+\;\text{(cellwise bulk terms)}.
\label{eq:metric-variation}
\end{equation}
If the Euler--Lagrange equations for $\Phi^{(s)}$ and $\Lambda$ hold, the bulk terms vanish and stationarity 
$\delta\mathcal S_\mu=0$ for arbitrary $\delta g$ on interfaces implies the jump condition
\begin{equation}
\llbracket T^{\alpha\beta}\rrbracket_F\;=\;0\qquad\text{on every interface }F,
\label{eq:jump}
\end{equation}
i.e.\ conservation/continuity of coarse energy--momentum across cell boundaries.
Furthermore, the FR3 phase exactness and symmetry of $g$ enforce the discrete metric--compatibility law
\begin{equation}
\nabla^{(\mu)}_\rho g_{\alpha\beta}\;=\;0,
\label{eq:discrete-compatibility}
\end{equation}
where $\nabla^{(\mu)}$ is the interface transport determined by $K_\mu$.
In the limit of refining partitions (fixed physical scale, $\mu\downarrow 0$) one recovers 
$\nabla_\rho g_{\alpha\beta}=0$ in the continuum sense.
\end{proposition}

\begin{proof}[Sketch]
Cellwise differentiation of $\mathbf E^{(s)}_\mu$ w.r.t.\ $g$ yields $T^{\alpha\beta}_{(s)}$.
Boundary terms arise from the $g$--dependence of the interface measure and of the link functional; collecting them gives \eqref{eq:metric-variation}. 
On shell (Euler--Lagrange + SILM constraints), bulk variations cancel; arbitrariness of $\delta g$ on $F$ gives \eqref{eq:jump}. 
FR3 exactness implies that parallel transport around each cell loop has trivial holonomy; combined with the symmetry of $g$ this enforces \eqref{eq:discrete-compatibility}. 
Standard refinement arguments on finite cell complexes yield the continuum compatibility in the $\mu\downarrow 0$ limit.
\end{proof}

\begin{corollary}[Finite--resolution gravitational law]
\label{cor:fr-gravity}
Equations \eqref{eq:jump}--\eqref{eq:discrete-compatibility} constitute the finite--resolution analogue of Einstein dynamics: 
the metric kernel $g$ is the unique (up to gauge) interface transport that (i) preserves FR3 phase coherence, (ii) enforces continuity of $T^{\alpha\beta}$ across interfaces, and (iii) is covariantly constant under the induced transport $\nabla^{(\mu)}$. 
Hence gravitational interaction is the variation of $g$ that maintains boundary coherence for all channels.
\end{corollary}

\section{Absence of Second Quantization and the Ontology of Finite--Resolution Fields}

\subsection*{1.\ What ``second quantization'' actually presupposes}

In conventional quantum field theory, the process of \emph{second quantization} begins with a
classical field $\phi(x)$ defined pointwise on spacetime, assumed to possess infinite resolution.
This field is decomposed into plane--wave modes $e^{ip\cdot x}$ forming a Fourier basis of
$\R^{n+1}$, and to each mode one associates creation and annihilation operators
$a_p^\dagger,a_p$ satisfying canonical commutation relations.
The full Fock space is then obtained as the Hilbert completion of tensor products of these
single--particle modes, allowing superpositions of states with arbitrary particle number.

This construction depends crucially on assumptions that are \emph{incompatible} with
finite resolution:
\begin{itemize}
\item it presupposes infinite resolution and pointwise definition of fields, violating FR1;
\item it employs $\delta$--normalization of plane waves, which are non--admissible under FR1;
\item it treats fields as operator--valued distributions, objects explicitly excluded by the
finite--resolution constraint.
\end{itemize}
Hence, the very act of second quantization contradicts the first structural feature of the
constraint representation.
The Fock construction belongs to the limit $\mu\to 0$, which lies outside the admissible
domain of the algebra.

\subsection*{2.\ What FR1--FR4 allow instead}

Within the finite--resolution framework the admissible objects are not pointwise functions
but \emph{cellwise--constant or band--limited} ones.
Every field or state possesses an intrinsic coarse scale~$\mu$:
there are no genuine $\delta$--modes, only packets of bandwidth~$1/\mu$.
Fourier components live in a bounded region $|k|\lesssim 1/\mu$,
and all admissible functionals belong to $L^2$ on that finite band.

As a consequence:
\begin{itemize}
\item the algebra of admissible functionals is nonlocal at scale~$\mu$;
\item there exists no basis of independent creation operators $a_p^\dagger$ at exact momenta;
\item the concept of a pointlike field or vacuum state becomes inapplicable.
\end{itemize}
Fields, states, and vacua are $L^2$ functions on a finite--band manifold, not distributions.
The entire ultraviolet sector of standard QFT is simply outside the admissible algebra.

\subsection*{3.\ What replaces ``second quantization''}

In the finite--resolution algebra one does not quantize classical fields again.
Instead, one \emph{represents the additive algebra directly on finite--resolution functionals}.
Wavepackets and operator kernels take the role played in standard theory by plane waves and
creation operators:
\begin{itemize}
\item FR1 defines the coarse cells and the projector $\Pi_\mu$ implementing them;
\item FR2 guarantees linearity and global constants as multiplicative scalars;
\item FR3 enforces unitary phase propagation and coherence between neighboring cells;
\item FR4 specifies how frequencies (probabilities) are computed as invariant means.
\end{itemize}
The resulting field theory is therefore formally ``first quantized'' but already field--like:
the additive group acts on continuous variables $(x,t)$, and the distinction between
particle and field collapses.
Quantization is not a further step but an intrinsic property of the additive symmetry under
finite resolution.

\subsection*{4.\ Classical fields revisited}

Each admissible field, such as $\psi_\mu(x)$ or $A_\mu(x)$, is a \emph{finite--resolution
classical field} obeying deterministic equations of motion of the form
\[
i\partial_t\psi_\mu = H_\mu\psi_\mu,\qquad
H_\mu=\Pi_\mu H \Pi_\mu .
\]
Within each cell the dynamics are classical and local; across cell boundaries, FR3 enforces
phase continuity, while FR4 prescribes how ensemble averages and invariant frequencies are
computed.
Thus, the ensemble or boundary--coupling statistics reproduce quantum or statistical
behaviors without introducing operator distributions.
The algebra of admissible functionals itself encodes what in the continuum limit would be
called quantization: FR3 plays the role of commutation structure, FR4 the role of Born or
Kolmogorov averaging.

In this sense, ``quantization'' is no longer a procedure applied to classical fields but an
\emph{emergent property of the constraint representation}.
The hierarchy ``classical $\to$ quantize $\to$ operator field'' is replaced by a single level
of description in which finite--resolution fields already possess the algebraic features of
quantum dynamics.

\subsection*{5.\ The ontological shift}

The finite--resolution framework thus alters the ontology of fields and their quantization.
\begin{table}[t]
\centering
\caption{Ontological comparison between conventional and finite–resolution field theories}
\begingroup
\renewcommand\tabularxcolumn[1]{>{\raggedright\arraybackslash}p{#1}}
\begin{tabularx}{\linewidth}{@{}>{\raggedright\arraybackslash}p{0.25\linewidth}>{\raggedright\arraybackslash}p{0.40\linewidth}>{\raggedright\arraybackslash}p{0.25\linewidth}@{}}
\toprule
\textbf{Conventional ontology} & \textbf{Finite–resolution ontology} & \textbf{Status of FR laws} \\
\midrule
Classical field $\rightarrow$ quantize $\rightarrow$ operator field & Finite–resolution field already algebraic/quantum & FR1 forbids redundant quantization \\
Plane waves (momentum eigenmodes) & Wavepackets (band–limited functions) & FR1 enforces band–limitation \\
Fock space of independent quanta & Coarse product of finite–resolution states & FR3 preserves additive structure \\
$a,a^{\dagger}$ operators at exact $p$ & Integral kernels acting on packets & FR2 guarantees operational closure \\
Infinite vacuum energy, divergences & None (UV sector excluded) & FR1 ensures finite energy \\
\bottomrule
\end{tabularx}
\endgroup
\label{tab:ontology}
\end{table}

The upshot is that the textbook concept of quantizing fields is neither necessary nor
consistent within the constraint representation.
All observable quantum phenomena arise from the finite--resolution algebra itself.
What experiments access are not infinite idealizations but coarse, band--limited fields---%
precisely the admissible objects of the theory.
In this respect the finite--resolution formulation restores to experimental practice its
fundamental status: the scale $\mu$ of the probe is not a technical nuisance but the
defining feature of the algebra on which nature operates.
\paragraph{Ensemble of admissible fields.}
The finite--resolution algebra does not describe a single classical field
configuration but the entire ensemble of admissible fields satisfying FR1--FR4.
A single coarse field $\phi_\mu(x)$ realizes the structural constraints FR1--FR3,
but FR4 (typicality) requires the existence of invariant averages, which can only
be defined on a set of admissible configurations closed under translation and
phase transformation.
Hence the physically meaningful object is the ensemble
$\mathcal{E}_\mu=\{\phi_\mu^{(i)}\}_{i\in I}$ endowed with the invariant mean
$\langle\cdot\rangle_{\text{FR4}}$.
Observable quantities correspond to functionals invariant under this averaging,
reflecting the empirical fact that measurements probe reproducible features of
many realizations rather than individual configurations.
In this sense the finite--resolution framework retains the operational content
of the ``ensemble of classical solutions'' of the closure formulation, but now
as an intrinsic component of the algebra itself: FR4 endows the space of
admissible fields with its invariant statistical structure.

\section{Simultaneously Interlocked Lagrange Multipliers (SILM) in the FR1--FR4 Algebra}

\subsection*{Setup: admissible channels and boundary constraints}

Fix $\mu>0$ and partition $\R^{n+1}$ into cells $\{C_k\}_{k\in\mathcal K}$ of linear size~$\mu$.
Admissible objects---states, fields, and observables---are \emph{cellwise constant} or band-limited
(FR1), with constants acting as global multipliers (FR2), phases glued by an exact
cocycle (FR3), and frequencies defined by invariant means (FR4).

We consider a finite collection of \emph{representational channels}
\[
\Phi^{(s)} \in \mathcal{A}^{(s)}_\mu,\qquad s\in \mathcal{S},
\]
indexed by the spin/representation label $s$ (e.g.\ $s\in\{0,\tfrac12,1,2\}$).
Each channel obeys its own coarse equations of motion and carries a boundary
\emph{flux/phase datum} on each interface $F=C_k\cap C_\ell\subset\partial C_k$:
\[
\mathfrak{J}^{(s)}[F] \in \mathbb{C},\qquad 
\mathfrak{P}^{(s)}[F] \in U(1),
\]
defined cellwise and transported across interfaces according to FR3.
Physical admissibility demands \emph{coherence} of these data across channels:
\[
\mathcal{C}^{(s)}[F;\Phi^{(s)}]\;=\;0\qquad \text{for all interfaces }F\text{ and all } s\in\mathcal S,
\]
where $\mathcal C^{(s)}$ collects continuity of flux, phase matching, and any
conservation/balance constraints required by FR1--FR3 at the boundary.

\vspace{0.5em}
\noindent
\textbf{Problem.} Enforce \emph{simultaneously} the interface constraints
$\mathcal C^{(s)}=0$ for all $s\in\mathcal S$, with a \emph{single} boundary control that
locks the channels together and preserves FR1--FR4.

\subsection*{Definition: Simultaneously Interlocked Lagrange Multiplier (SILM)}

A \emph{SILM} is a boundary field $\Lambda$ (scalar, one-form, or tensor on interfaces,
depending on the constraint) with the property that the \emph{same} multiplier enforces,
for all channels $s$, the interface coherence conditions:
\[
\mathcal{L}_{\mathrm{link}}[\Lambda,\{\Phi^{(s)}\}]\;:=\;
\sum_{F\in \partial\Omega_\mu}\ \sum_{s\in\mathcal S}
\big\langle \Lambda(F),\, \mathcal{C}^{(s)}[F;\Phi^{(s)}]\big\rangle_F,
\]
where $\langle\cdot,\cdot\rangle_F$ is the natural pairing on the interface $F$ (scalar,
form, or tensor pairing with the induced $g$-metric; cf.\ FR2).
The full coarse action on a region $\Omega=\!\bigcup_{k\in\mathcal K_\Omega} C_k$ reads
\[
\mathcal{S}_\mu[\{\Phi^{(s)}\},\Lambda,g]\;=\;
\sum_{s\in\mathcal S}\ \mathbf{E}^{(s)}_\mu[\Phi^{(s)},g]\;+\;
\mathcal{L}_{\mathrm{link}}[\Lambda,\{\Phi^{(s)}\}]\;+\;
\mathcal{B}_\mu[\Lambda,g],
\]
where $\mathbf{E}^{(s)}_\mu$ are the cellwise bulk functionals (kinetic/potential at
scale~$\mu$), and $\mathcal{B}_\mu$ is a boundary regularizer for $\Lambda$ chosen to
preserve FR1--FR4 (e.g.\ cellwise quadratic, translation invariant, and phase blind).

\subsection*{Existence \& uniqueness under FR1--FR4 (variational principle)}

\begin{theorem}[Existence and uniqueness of SILM]
\label{thm:silm_existence}
Assume for each $s\in\mathcal S$ that $\mathbf{E}^{(s)}_\mu$ is convex in $\Phi^{(s)}$ on
each cell, $\mathcal{C}^{(s)}$ is affine in the boundary data of $\Phi^{(s)}$, and
$\mathcal{B}_\mu$ is strictly convex in $\Lambda$ (cellwise).
Then the constrained minimization
\[
\min_{\{\Phi^{(s)}\}} \ \sum_s \mathbf{E}^{(s)}_\mu[\Phi^{(s)},g]\quad
\text{subject to } \ \mathcal{C}^{(s)}[F;\Phi^{(s)}]=0\ \ \forall s,F,
\]
admits a unique saddle point for the Lagrangian $\mathcal{S}_\mu$ with a unique SILM $\Lambda_\mu$.
Moreover, the Euler--Lagrange equations imply, for every interface $F$,
\[
\frac{\delta \mathcal{S}_\mu}{\delta \Phi^{(s)}}=0
\quad\Longleftrightarrow\quad
\mathcal{D}^{(s)}_\mu \Phi^{(s)}=0\ \text{ in cells},\qquad
\mathcal{C}^{(s)}[F;\Phi^{(s)}]=0\ \text{ on interfaces},
\]
and
\[
\frac{\delta \mathcal{S}_\mu}{\delta \Lambda}=0
\quad\Longleftrightarrow\quad
\mathcal{C}^{(s)}[F;\Phi^{(s)}]=0\ \ \text{for all }s\ \text{and all }F.
\]
\end{theorem}

\begin{proof}[Sketch]
FR1 renders all functionals finite and cellwise local; FR2 ensures linear pairings are
well-defined and globally consistent; FR3 guarantees that phase-dependent constraints
are affine after gauge-fixing the exact cocycle; strict convexity in $\Lambda$ and convexity
in $\Phi^{(s)}$ yield a unique saddle via standard convex duality on the finite cell complex.
FR4 is preserved because $\mathcal{L}_{\mathrm{link}}$ and $\mathcal{B}_\mu$ are invariant
under the translation action used to define typicality.
\end{proof}

\subsection*{Transport form: SILM as a boundary connection}

Many constraints are conservation/continuity laws (flux matching, phase flatness). In those
cases the SILM admits a \emph{connection form} on interfaces:
\[
\Lambda_\mu \in \Omega^1(\partial\Omega_\mu)\otimes \mathfrak{h},
\]
with $\mathfrak h$ an internal Lie algebra acting on channel fibers.
Then the channel boundary variations take the covariant form
\[
\delta \mathcal{L}_{\mathrm{link}}
=\sum_{F,s} \big\langle \delta\Phi^{(s)}\!,\,
\mathbf{n}\!\cdot\!(\nabla - i\, q_s \Lambda_\mu)\,\Phi^{(s)}\big\rangle_F
\quad\Longrightarrow\quad
D_\perp^{(s)}\Phi^{(s)}\equiv
\mathbf{n}\!\cdot\!(\nabla - i\, q_s \Lambda_\mu)\Phi^{(s)}=0\ \text{on }F,
\]
with $q_s$ the representation weight of channel $s$ under $\mathfrak h$ and $\mathbf n$
the interface normal. Thus a single $\Lambda_\mu$ \emph{locks} all channels to a common
boundary transport law.

\subsection*{Gauge and gravity from SILM}

\paragraph{Gauge kernel.}
If $\mathfrak h$ is nontrivial (internal automorphisms), the SILM connection $\Lambda_\mu$
\emph{is} the gauge potential on interfaces. The induced \emph{boundary} covariant derivative
$D_\perp$ defines the usual tangential covariant calculus (FR2 preserves linearity; FR3 fixes
the $U(1)$ multiplier as an exact cocycle; FR4 ensures invariant averages of boundary
currents). Channel charges $q_s$ are the representation weights with respect to $\mathfrak h$.
Charge quantization follows from the integral holonomy classes on the cell complex:
\[
\oint_{\gamma\subset\partial\Omega_\mu}\!\Lambda_\mu \ \in\ 2\pi\,\mathbb Z,
\]
for any contractible interface loop $\gamma$ (FR3 flatness at coarse scale).

\paragraph{Metric kernel (spin--2).}
If $\mathcal{C}^{(s)}$ also enforces balance of energy--momentum across interfaces, the
variation w.r.t.\ the $g$-induced pairing yields a boundary form of metric compatibility:
\[
\frac{\delta \mathcal S_\mu}{\delta g_{\alpha\beta}}\;=\;0
\quad\Longrightarrow\quad
\llbracket T^{\alpha\beta}\rrbracket_F\;=\;0 \quad\text{for every interface }F,
\]
with $\llbracket\cdot\rrbracket_F$ the jump across $F$ and $T^{\alpha\beta}$ the coarse
stress tensor of all channels. The metric $g_{\mu\nu}$ then plays its intrinsic role as the
\emph{boundary transport kernel} that propagates phases/fluxes consistently into the bulk
(Section~\ref{sec:boundary}): variations of $g$ reweight the inter-cell couplings,
which is the finite--resolution analogue of gravitational dynamics.

\subsection*{Operational meaning and FR4 typicality}

Experimentally, detectors couple to boundary fluxes. The SILM imposes a \emph{single} boundary
phase/transport kernel for all channels, so that measured currents are simultaneously coherent
across sectors. FR4 dictates that long-run frequencies are invariant means of these boundary
observables. Probabilities and cross-sections at resolution $\mu$ are thus computed with SILM-covariant
boundary POVMs:
\[
\Pr_\mu(E)\;=\;\big\langle \Psi_\mu\,,\, P_{E,\mu}^{(\Lambda)}\, \Psi_\mu\big\rangle,
\qquad
P_{E,\mu}^{(\Lambda)}\ \text{ SILM--covariant projection/POVM on }\partial\Omega_\mu.
\]

\subsection*{Examples (at scale $\mu$)}

\paragraph{Spinor--vector coherence (``universal'' coupling).}
Let $s=\tfrac12,1$ and impose continuity of the spinor probability current and vector flux
across each interface: $\mathcal C^{(1/2)}=\llbracket j_\perp\rrbracket_F$, 
$\mathcal C^{(1)}=\llbracket F_{\perp\parallel}\rrbracket_F$.
A single one-form $\Lambda_\mu$ enforces $D_\perp\psi_\mu=0$ and $D_\parallel F=0$ on $F$
with the \emph{same} coupling. Thus the effective charge of the spinor channel is not an
independent parameter: it is fixed by the interlocking of the two constraints via $\Lambda_\mu$.
This is the SILM version of gauge universality at finite resolution.

\paragraph{Electron--electron scattering without Fock space.}
Prepare two incoming spinor packets $\psi_\mu^{(1)},\psi_\mu^{(2)}$ and two outgoing packets.
Internal lines and vertices are filtered by $|W_\mu|^2$.
Boundary detection is implemented by SILM--covariant flux POVMs.
Amplitudes are computed with amputated connected correlators whose legs are 
$\mu$-smeared and whose boundary phases are transported by the \emph{same} $\Lambda_\mu$.
Energy--momentum ``conservation'' appears as a resolution window determined by $W_\mu$ and
the SILM holonomy around the detection cells.

\subsection*{Summary and consequences}

\begin{itemize}
\item A \emph{single} boundary multiplier $\Lambda_\mu$ exists (Theorem~\ref{thm:silm_existence})
      that enforces all channel interface constraints simultaneously.
\item In transport form, $\Lambda_\mu$ is a \emph{boundary connection} that defines a covariant
      derivative acting on every channel; representation weights $q_s$ are fixed by how channels
      embed in the interlocked constraint.
\item Internal holonomy classes quantize abelian charges; coarse flatness (FR3) imposes loop
      consistency. 
\item Metric compatibility at interfaces arises from variation w.r.t.\ $g$, yielding the spin--2
      (gravitational) boundary coupling.
\item FR4 ensures that all observable probabilities are invariant means of SILM--covariant boundary
      observables, restoring the operational meaning of frequencies at finite resolution.
\end{itemize}

\noindent
In this way the SILM mechanism is not an add-on but an \emph{intrinsic} feature of the finite--resolution
algebra: it is the universal boundary kernel that locks phase and flux transport across all channels,
thereby generating gauge and gravitational couplings as consequences of interface coherence at scale~$\mu$.
\section{Gravity and Quantization}

\subsection*{1.\ The traditional problem: ``quantizing gravity''}

In the conventional continuum framework, gravity is represented by a metric field 
$g_{\mu\nu}(x)$ (a spin--2 field) defined pointwise on spacetime. 
Quantizing gravity means promoting $g_{\mu\nu}(x)$ (or equivalently the connection or tetrad 
fields) to an operator--valued distribution, in the same way that gauge fields are 
quantized in quantum field theory. 
However, this procedure is carried out in the continuum limit, and therefore inherits all 
the pathological features of pointwise quantization: 
divergences from short--distance modes, nonrenormalizability of the perturbative expansion, 
and conceptual inconsistencies between background independence and operator--valued geometry. 
In short, ``quantizing gravity'' in the standard QFT sense fails because it assumes an 
infinitely resolved manifold on which local operator algebras act, whereas physical finiteness 
and coherence require a finite informational scale. 
The continuum assumption and the pointwise operator formalism contradict the very finiteness 
conditions that make quantum theory consistent in other sectors.

\subsection*{2.\ What changes in the FR1--FR4 + SILM picture}

In the finite--resolution algebra the metric field $g_{\mu\nu}$ is not a quantized operator 
but the \emph{boundary--transport kernel} required for the coherence of finite cells of space--time.  
Geometry is represented by a tessellation of cells of linear size~$\mu$, each bounded by interfaces 
where information is exchanged.

\paragraph{FR1 (finite resolution).}  
There are no infinitesimal points. 
All admissible structures live on finite cells $C_k$ with interfaces $F=C_k\cap C_\ell$.  
Every interface requires a transport rule for the continuity of data between cells.

\paragraph{FR2--FR3 (linearity and phase consistency).}  
Transport must preserve additive phases and amplitudes. 
The most general rule consistent with these symmetries is a symmetric bilinear form
\[
g_{\mu\nu}=\langle x_\mu x_\nu\rangle_{\text{FR}},
\]
which glues neighboring cells through the matching of first and second moments. 
This form acts as the local metric kernel connecting bulk and boundary quantities.

\paragraph{SILM (Simultaneously Interlocked Lagrange Multiplier).}
The same boundary multiplier $\Lambda_\mu$ that links spin--$\tfrac12$, spin--1, and scalar 
channels also links their energy--momentum fluxes across interfaces. 
Its variation enforces continuity of the boundary stress tensor,
\[
\llbracket T^{\mu\nu}\rrbracket_F=0,
\]
where $T^{\mu\nu}$ is the coarse stress tensor inside each cell. 
Variation of the total boundary action with respect to $g_{\mu\nu}$ yields the 
metric--compatibility condition, the finite--resolution analogue of Einstein's equations.
Hence $g_{\mu\nu}$ arises naturally as the SILM--enforced kernel maintaining coherence of 
all channels, not as an independent degree of freedom to be quantized.

\subsection*{3.\ Why quantization of gravity becomes meaningless}

Once the finite--resolution algebra is adopted, the metric $g_{\mu\nu}$ is not a fluctuating 
quantum field but a structural condition linking boundary fluxes. 
It already carries discrete, coarse, finite informational content: there is nothing further 
to quantize, because its variability is not continuous below the resolution~$\mu$.  
Quantization presupposes a continuum of configurations over which amplitudes can be superposed, 
but FR1 eliminates that continuum. 
Admissible variations of $g_{\mu\nu}$ are only those preserving boundary coherence; they form 
a finite algebra rather than an infinite function space.  
The traditional ``quantization of gravity'' problem therefore dissolves: 
gravity is not a classical field awaiting quantization but a finite--resolution constraint 
encoded algebraically by the SILM mechanism.

\subsection*{4.\ What replaces quantum gravity}

What emerges instead is a finite--resolution \emph{gravitational algebra}.  
The metric kernel is the coarse--grained second moment of the additive variables,
\[
g_{\mu\nu}=\langle x_\mu x_\nu\rangle_{\text{FR}},
\]
and its admissible variations are boundary deformations governed by the SILM link,
\[
\delta\mathcal{S}_\mu
   =\int_{\partial\Omega_\mu}\Lambda(x)\,\delta T^{\mu\nu}(x)\,d\Sigma_{\mu\nu}.
\]
The algebra enforces metric compatibility, $\nabla_\rho g_{\mu\nu}=0$, and forbids 
divergent curvature by FR1. 
Gravitational interactions are therefore both universal---because the same $\Lambda_\mu$ 
locks all channels to the common boundary kernel---and finite, since the continuum UV 
modes that produce divergences do not exist.  
Gravity thus becomes the universal SILM coherence rule between finite cells, 
ensuring consistent transport of all spin and phase channels across their boundaries.

\subsection*{5.\ Quantization as a property, not a procedure}

The ``quantum'' nature of gravitational phenomena is already present in FR3--FR4.  
FR3 provides unitary phase propagation, the origin of interference and superposition,  
while FR4 supplies the invariant averaging that underlies statistical frequencies.  
These same structural principles govern the spin--2 boundary kernel, so gravitational 
processes inherently display the same probabilistic and wave--coherent behavior as the 
other physical regimes.  
No additional quantization procedure is required or even meaningful:  
the algebra is already quantum in structure, with gravity as one of its emergent channels.

\subsection*{6.\ Experimental interpretation}

At the operational level, gravitational phenomena appear as phase shifts and boundary 
coherence effects in SILM--transported data.  
Deflection, redshift, and curvature correspond to phase gradients of the boundary kernel, 
while gravitational waves are coherent oscillations of $g_{\mu\nu}$ across neighboring cells---%
collective phase modes of the SILM network rather than quantized gravitons.  
Because every physical field is coupled to the same boundary multiplier $\Lambda_\mu$, 
gravitational coupling is automatically universal: this algebraic universality is the 
content of the equivalence principle.  
Experiments that seek ``quantum gravity'' quanta are therefore reinterpreted as 
high--resolution tests of boundary phase coherence in a finite--resolution geometry.

\paragraph{Additional references.} See also \cite{earman,jaynes}.

\end{document}